\documentclass[11pt,reqno]{amsart}
\usepackage[top=2.54cm, bottom=2.54cm, left=2.8cm, right=2.8cm]{
geometry}
\usepackage{dsfont, amssymb,amsmath,amscd,latexsym, amsthm, amsxtra,amsfonts, extarrows}
\usepackage[all]{xy}
\usepackage[active]{srcltx}
\usepackage[round]{natbib}
\usepackage{bbm}
\usepackage{enumerate}
\usepackage{mathrsfs}
\usepackage{graphicx}
\usepackage{comment}
\usepackage{mathtools}

\usepackage{tikz}
\usetikzlibrary{calc,arrows}
\usepackage{verbatim}
\usepackage{graphicx}
\usepackage{subfigure}

\newtheorem{theorem}{Theorem}[section]
\newtheorem{corollary}[theorem]{Corollary}
\newtheorem{proposition}[theorem]{Proposition}
\newtheorem{lemma}[theorem]{Lemma}

\theoremstyle{definition}

\renewcommand{\theequation}{\arabic{section}.\arabic{equation}}

\theoremstyle{definition}

\theoremstyle{definition}
\newtheorem{remark}{Remark}
\theoremstyle{definition}
\newtheorem{assumption}{Assumption}

\renewcommand{\epsilon}{\varepsilon}

\usepackage[pdfstartview=FitH, bookmarksnumbered=true,bookmarksopen=true, colorlinks=true, pdfborder=001, citecolor=blue, linkcolor=blue,urlcolor=blue]{hyperref}
\usepackage{graphics}
\graphicspath{{figures/}}
\usepackage{xpatch}
\usepackage{xcolor}

\makeatletter
\ExplSyntaxOn
\cs_new:Npn \bibColoredItems #1#2
{
	\clist_map_inline:nn {#2} { \cs_new:cpn {bib@colored@##1} {#1} }
}
\ExplSyntaxOff

\newcommand\bib@setcolor[1]{%
	\ifcsname bib@colored@#1\endcsname
	\expandafter\color\expandafter{\csname bib@colored@#1\endcsname}
	\else
	\normalcolor
	\fi
}

\xpatchcmd\@bibitem
{\item}
{\bib@setcolor{#1}\item}
{}{\fail}

\xpatchcmd\@lbibitem
{\item}
{\bib@setcolor{#2}\item}
{}{\fail}
\makeatother

\begin{document}
\makeatletter
\def\@setauthors{%
\begingroup
\def\thanks{\protect\thanks@warning}%
\trivlist \centering\footnotesize \@topsep30\p@\relax
\advance\@topsep by -\baselineskip
\item\relax
\author@andify\authors
\def\\{\protect\linebreak}%
{\authors}%
\ifx\@empty\contribs \else ,\penalty-3 \space \@setcontribs
\@closetoccontribs \fi
\endtrivlist
\endgroup} \makeatother
 \baselineskip 19pt
 \title[{{\tiny Consumption-portfolio choice with preferences for liquid assets}}]
 {{\tiny
Consumption-portfolio choice with preferences for liquid assets}} \vskip 10pt\noindent
\author[{\tiny  Guohui Guan, Jiaqi Hu, Zongxia Liang}]
{\tiny {\tiny  Guohui Guan$^{a,b,*}$, Jiaqi Hu$^{c,\dag}$, Zongxia Liang$^{c,\ddag}$}
 \vskip 10pt\noindent
{\tiny ${}^a$Center for Applied Statistics, Renmin University of China, Beijing 100872, China
\vskip 10pt\noindent\tiny
${}^b$School of Statistics, Renmin University of China, Beijing 100872, China
\vskip 10pt\noindent
${}^c$Department of Mathematical Sciences, Tsinghua
University, Beijing 100084, China
}\noindent
\footnote{{$^*$ {\bf e-mail}: guangh@ruc.edu.cn}}\noindent
\footnote{{$^\dag $ Corresponding author, \  {\bf e-mail}: hujq20@mails.tsinghua.edu.cn }}  \noindent
 \footnote{{$\ddag $ {\bf e-mail}: liangzongxia@mail.tsinghua.edu.cn}}}
\numberwithin{equation}{section}
\maketitle
\noindent
\begin{abstract}
This paper investigates an infinite horizon, discounted, consumption-portfolio problem in a market with one bond, one liquid risky asset, and one illiquid risky asset with proportional transaction costs. We consider an agent with liquidity preference, modeled by a Cobb-Douglas utility function that includes the liquid wealth. We analyze the properties of the value function and divide the solvency region into three regions: the buying region, the no-trading region, and the selling region, and prove that all three regions are non-empty. We mathematically characterize and numerically solve the optimal policy and prove its optimality. Our numerical analysis sheds light on the impact of various parameters on the optimal policy, and some intuition and economic insights behind it are also analyzed. We find that liquidity preference encourages agents to retain more liquid wealth and inhibits consumption, and may even result in a negative allocation to the illiquid asset. The liquid risky asset not only affects the location of the three regions but also has an impact on consumption. However, whether this impact on consumption is promoted or inhibited depends on the degree of risk aversion of agents.
\end{abstract}
\vskip10pt
\setcounter{equation}{0}

\section{{{\bf Introduction}}}

Since the seminal work of  \cite{Merton1969, merton1975optimum}, the investment-consumption problems in continuous time have been extensively studied. In a perfect, frictionless market, the optimal policy for a constant relative risk aversion (CRRA) investor is to maintain a constant fraction of total wealth in each asset and consumption.  \cite{MAGILL1976} introduced proportional transaction costs to the Merton's model and discovered the existence of a non-transaction region. Despite the attention this problem has received in the literature, many studies on transaction costs overlook the impact of liquidity preference on decision-making processes. This paper aims to fulfill such a gap.

The concept of liquidity preference was introduced by \cite{k36} as one of his three psychological laws. Initially, liquidity preference referred mainly to a preference for money. People tend to prefer money because of its flexibility in use, and they would rather sacrifice some potential profits than hold onto a certain amount of money. One way to model money preference is by incorporating it into the utility function. While a series of macroeconomic literature, starting with \cite{sidrauski1967rational} and \cite{brock1974money}, have used ``money in the utility function" for research, most studies on consumption-portfolio decisions have overlooked the role of money. \cite{kraft2019consumption} first considered ``money in the utility function" in the consumption-portfolio problem and demonstrated that in an infinite-horizon setting, liquidity preference leads to a decrease in investment in risky assets. However, their work only considered a Black-Scholes market and neglected transaction costs. Compared to \cite{kraft2019consumption}, this paper deals with three assets: a risk-free bond, a liquid risky asset, and an illiquid risky asset, where the risk-free bond and the liquid risky asset together constitute liquid wealth, and trading in the illiquid asset incurs proportional transaction costs. It is worth noting that the essence of money preference is liquidity preference. More than a hundred years ago, \cite{k36} equated money preference and liquidity preference, but now that various financial derivatives are abundant, liquidity preference should not simply refer to money preference. To capture this liquidity preference, we add liquid wealth to the utility function. 

Traditional economics holds that individual happiness comes from consumption. We believe that this is one-sided and that holding liquid wealth may also be an important factor affecting happiness. Because of different cultures and national conditions, different countries and regions have different willingness to consume and hold liquid wealth. For example, China is a country with a very strong concept of household savings. Over the years, the growth of household savings in China has been higher than the growth of GDP, and in recent years, it has shown rapid growth. Holding a certain amount of liquid wealth can not only prevent future uncertain events, but also bring more possibilities to future consumption, which gives Chinese residents a sense of security, which Chinese residents also attach great importance to. Residents of countries such as China, Singapore, South Korea and Ireland also value the sense of security that comes from holding liquid assets. The sense of security brought by holding liquid assets and the satisfaction brought by consumption together constitute the happiness of residents, and both need to be equally important reflected in the utility function.

This paper provides a comprehensive investigation  over the liquidity preference of assets. Specifically, we divide liquid wealth into a risk-free bond and a liquid risky asset. Moreover, we introduce liquidity preference and incorporate ``liquid wealth in the utility function" into the consumption-portfolio problem with proportional transaction costs. Liquidity enters the utility function indirectly because it saves time for the agent in conducting transactions in liquid assets. The agent's utility function is given by $$U(c,x)=\frac{1}{p}\left(c^{\theta}x^{1-\theta}\right)^p,$$ where $c$ is the consumption rate, $x$ is the liquid wealth, $p$ is the CRRA risk aversion parameter with $p<1$, and $\theta\in\left(0,1\right)$ is the liquidity preference parameter.

The continuous-time optimal investment problem with transaction costs has been approached using various methods, including stochastic control and viscosity theory (see \cite{Davis1990}, \cite{Sherve1994}, \cite{kabanov2004geometric}, \cite{de2016consumption}), the dual approach and shadow prices (see \cite{Kallsen2010}, \cite{Guasoni2013}, \cite{choi2013shadow}), martingale approach (see \cite{cvitanic1996hedging}), and numerical solutions (see \cite{gennotte1994investment}, \cite{muthuraman2006multidimensional}). Additionally, \cite{constantinides1986capital} discussed the effect of transaction costs on liquidity premium, while \cite{soner2013homogenization} carried out perturbation analysis for small transaction costs. However, most of the existing models with transaction costs only involve a single risky asset, and those with multiple assets are notably harder to analyze. In this paper, we focus on the consumption-investment problem with liquidity preference in the infinite-horizon case and a multi-asset setting. On the computational side in a multi-asset setting, \cite{akian1996investment} considered $n$ uncorrelated risky assets with proportional transaction costs and solved the variational inequality by using a numerical algorithm based on policies, iterations, and multi-grid methods. \cite{chen2013characterization} further considered $n$ correlated risky assets with proportional transaction costs and paid more attention to the shape and location of the non-trading region. This work also considers multiple risky assets, specifically a risk-free bond, a liquid risky asset, and an illiquid risky asset. Risky assets trade continuously, and their returns are potentially correlated. The coexistence of liquid and illiquid assets is more in line with the real financial environment, as reflected in recent literature. For example, \cite{bichuch2018investing} discussed the interaction between liquid and illiquid assets while maximizing the equivalent safe rate, and \cite{hobson2019multi} transformed the underlying HJB equation into a boundary value problem for a first-order differential equation and discussed the well-posedness of the problem. This paper extends \cite{hobson2019multi} by introducing liquidity preference, which leads to changes in the solvency region and the HJB equations. When $\theta=0$, the problem studied in this work reduces to the one in \cite{hobson2019multi}. We pay more attention to the shape and location of the non-trading region by analyzing the HJB equation of the problem. We find that even when liquidity preference is introduced into the utility function, the optimal decision can still be divided into three cone zones: buying, no-trading, and selling regions. Due to the introduction of liquidity preference, the solvency region is limited, so whether all three regions must exist needs to be re-studied. Given the assumption that the value function is known, we provide a characterization of the optimal policy and prove its optimality. Because liquid risk assets and illiquid risk exist at the same time, the HJB equation is inherently non-linear, and liquidity preference also makes it more complex. This puts forward higher requirements for the non-empty proof of the three regions as well as the characterization of the optimal strategy and the proof of its optimality. Finally, the influence of parameters on the boundary of no-trading region and the evolution of investment rate and consumption rate in no-trading region under the optimal policy are investigated by numerical analysis.

The main contributions of this paper are as follows.
\begin{itemize}
	\item Using the HJB equation, we analyze the properties of the value function to determine the shape of the optimal policy. Our analysis leads to a complete proof that all three trading regions (buying, no-trading, and selling regions) are non-empty, regardless of the value of $p$. Specifically, we prove that the three trading regions are non-empty both when $p<0$ and when $0<p<1$. The proof procedure for these two cases is different, but both require more elaborate estimation and calculation.
	
	\item  We provide a characterization of the optimal policy and rigorously prove its optimality. Due to the complexity of the HJB equation, the application of the stopping times is more flexible, and the analysis and estimation processes are more refined. The characterization of the optimal policy also reveals that the essence of the optimal policy has two points, one is the reflection behavior on the boundary of the no-trading region, and the other is the optimal investment and consumption in the no-trading region. This is also the focus of the numerical analysis discussion later.
	
	\item The risk aversion hypothesis is inherently inclusive of all cases where $p<1$, and according to \cite{MEHRA1985}, several studies have explored the possible realistic values of $p$, and the current literature suggests that $p<0$ is a reasonable and realistic range.  However, there are a few studies on the case of $p<0$ in the literature (there are certainly some, such as \cite{Sherve1994}), most papers only consider the case of $0<p<1$, because the case of $p<0$ is difficult and tedious both in mathematical proof and numerical calculation. Therefore, studying the case of $p<0$ in this paper is of great practical significance.
	
	\item We present numerical results based on \cite{Azimzadeh2016}. Our numerical analysis not only examines the influences of various parameters on the location of the three trading regions in the optimal policy but also investigates the evolutions of the investment ratio and consumption ratio in the no-trading region. We conduct numerical calculations and graphical analysis for both the $p>0$ and $p<0$ cases of the problem and obtain several interesting results. First, we find that the introduction of liquidity preference encourages agents to retain more liquid wealth and suppress consumption, but has little impact on the internal investment of liquid wealth. Second, we observe that the intersection of the no-trading region (or even the selling region) and the fourth quadrant may not be empty, but the intersection of the selling region and the fourth quadrant will be small if it exists. Third, we show that the introduction of liquidity risk assets has a certain impact on the location of the no-trading region, while the impact on consumption will depend on the level of the agents' risk aversion and may be reversed. We have made some reasonable intuition and economic insights on the phenomena and results manifested.
\end{itemize}

The paper is structured as follows: Section 2 presents the financial market and the optimization problem. Section 3 discusses the HJB equation and the properties of the value functions, including homotheticity, convexity, boundedness, and continuity, with a focus on boundary continuity. In Section 4, we derive the optimal policy, which consists of three trading regions, namely the buying region, no-trading region, and selling region, and prove that all three regions are non-empty.  Section 5 presents the numerical results and Section 6 is a conclusion.

\vskip 10pt
\section{\bf Market model and problem formulation}

\subsection{\bf The market model}
Let $\left(\Omega,\mathcal{F},\mathbb{P}\right)$ be a complete probability space with an augmented natural filtration $\left\{\mathcal{F}_t\right\}_{t\geq0}$ generated by two standard Brownian motions $\left(B^1,B^2\right)$ with correlation coefficient $\rho\in(-1,1)$. 

There are three assets in the economy: a risk-free bond $S^0$ earning a constant interest rate $r$ and two risky assets, one of which is a liquid risky asset $S^1$ while the other is an illiquid risky asset $S^2$ (unlike $S^1$, trading in $S^2$ incurs proportional transaction cost), which follow two bivariate geometric Brownian motions
\begin{eqnarray}
	\frac{dS^i_t}{S^i_t}=\left(r+\alpha_i\right)dt+\sigma_idB^i_t,\quad i=1,\ 2,\nonumber
\end{eqnarray}
where $r+\alpha_i\left(\alpha_i>0\right)$ is the mean rate of return of the risky asset $S^i$ and $\sigma_i$ represents the volatility coefficient of $S^i$, $i=1,2$.

Let $X_t$ denote the total value of the liquid wealth, i.e., the risk-free bond $S^0$ and the liquid risky asset $S^1$, at time $t$. Let $Y_t$ denote the value of the illiquid asset at time $t$. Define $\pi_t$ as the proportion of liquid wealth invested in the liquid risky asset at time $t$, and $c_t$ as the nonnegative consumption rate at time $t$. The processes $L_t$ and $M_t$ are both nonnegative, right-continuous, and nondecreaing, and they respectively represent the cumulative amounts of purchases and sales of the illiquid asset. In other words,
$$dL_t\ge 0,\quad dM_t\ge 0.$$
Trading the illiquid asset $S^2$ incurs a proportional transaction cost of $\lambda\in[0,1)$ on purchases and $\mu\in[0,1)$ on sales. A trading strategy is defined by a quadruple $\left(\pi,c,L,M\right)$, where $$c=\{c_t\}_{t\ge 0}\;,\;\pi=\{\pi_t\}_{t\ge 0}\;,\;L=\{L_t\}_{t\ge 0}\;,\;M=\{M_t\}_{t\ge 0}.$$ 
Suppose that transaction costs are paid in cash and consumption is from the cash account. Assuming that the strategy $\left(\pi,c,L,M\right)$ is self financing, 
$X:=\{X_t\}_{t\ge 0}$ and $Y:=\{Y_t\}_{t\ge 0}$ evolve according to the equations 
\begin{eqnarray}
	dX_t&=&\left[\left(r+\alpha_1\pi_t\right)X_t-c_t\right]dt + \sigma_1\pi_tX_tdB^1_t-dL_t+\left(1-\mu\right)dM_t,\quad X_{0-}=x,\label{xt}\\
	dY_t&=&\left(r+\alpha_2\right)Y_tdt+\sigma_2Y_tdB^2_t+\left(1-\lambda\right)dL_t-dM_t, \quad Y_{0-}=y.\label{yt}
\end{eqnarray}

\subsection{\bf Problem formulation}

We assume that the agent is not only concerned with the level of consumption, but also with the liquidity of their wealth. Specifically, the agents prefer to hold as much liquid wealth as possible while maximizing their expected utility from consumption. This is due to the fact that liquid wealth provide additional benefits due to their liquidity (see, for example, \cite{obstfeld1996foundations} and the references therein). Moreover, transaction costs may also incentivize the agent to hold more liquid wealth. The sense of security brought by holding liquid assets and the satisfaction brought by consumption together constitute the happiness of residents, and both need to be equally important reflected in the utility function. Consequently, the agent aims to maximize the objective function:
$$\mathbb{E}\left[\int_{0}^{\infty}e^{-\beta t}U\left(c_t,X_t\right)dt\right],$$ 
where $\beta>0$ is a constant discount rate. The agent's preference function incorporates a CRRA risk aversion parameter $p\in(-\infty,1)\setminus\{0\}$ and a liquidity preference parameter $\theta\in\left(0,1\right)$, and is given by $$U\left(c,x\right)=\frac{\left(c^{\theta}x^{1-\theta}\right)^p}{p},$$
which is a Cobb-Douglas function that captures the weight of consumption $c$ and liquid wealth $x$. A smaller risk aversion parameter $p$ indicates greater risk aversion by the agent, while a smaller liquidity preference parameter $\theta$ implies a higher proportion of liquid wealth in the utility function, and a greater desire to hold liquid wealth. Notably, when $\theta=1$, the utility function reduces to the classical consumption utility, which has been studied in \cite{hobson2019multi}. We always believe that liquidity preference exists and assume $\theta<1$. Due to the form of this utility function, it is necessary that $X\ge 0$, indicating that the agent should not only hold the illiquid wealth but always prefers to hold a certain amount of liquid wealth. We do not consider the degenerate case where $p=0$.

The net wealth of the agent after instantaneous liquidation of the illiquid asset is $$X_t+(1-\mu)Y_t^+-{Y_t^-\over 1-\lambda}.$$

We say that a trading strategy is admissible if the corresponding net wealth process and the liquid wealth process are all nonnegative. For this, following \cite{MAGILL1976}, we define the solvency region $$\Gamma=\{\left(x,y\right)\mid  x>0,x+\frac{y}{1-\lambda}>0\}$$
and we parition the boundary into
$$\partial_1\Gamma=\{\left(x,y\right)\mid x=0,y\ge0\}, \quad\partial_2\Gamma=\{\left(x,y\right)\mid x\ge0,x+\frac{y}{1-\lambda}=0\}.$$ 
For the sake of proof of continuity on $\partial_1\Gamma$, the region of the initial position $\left(X_{0-},Y_{0-}\right)$ can be expanded to $$\Lambda=\{\left(x,y\right)\mid x+\left(1-\mu\right)y>0,x+\frac{y}{1-\lambda}>0\}.$$
Compared with $\Lambda$, the solvency region $\Gamma$ imposes a restriction $x>0$. The agent with initial position $\left(X_{0-},Y_{0-}\right)$ on $\overline{\Lambda}\setminus\overline{\Gamma}=\{(x,y)\mid x<0,x+(1-\mu)y\ge0\}$ uses the illiquid asset to cover the short position in the liquid wealth and arrives at position $(0, Y_{0-}+\frac{X_{0-}}{1-\mu})$ on $\overline{\Gamma}$  instantaneously. In the closure of the expansion region $\overline{\Lambda}$, the net wealth after instantaneous liquidation is nonnegative. If the agent sells the illiquid asset $Y$ in the region $\overline{\Lambda}\setminus\overline{\Gamma}$, the liquid wealth $X$ is still nonnegative. The regions $\Gamma$ and $\Lambda$ are shown in Fig.~\ref{figure1}.

\begin{figure}[htbp]
	\centering
	\includegraphics[scale=0.8]{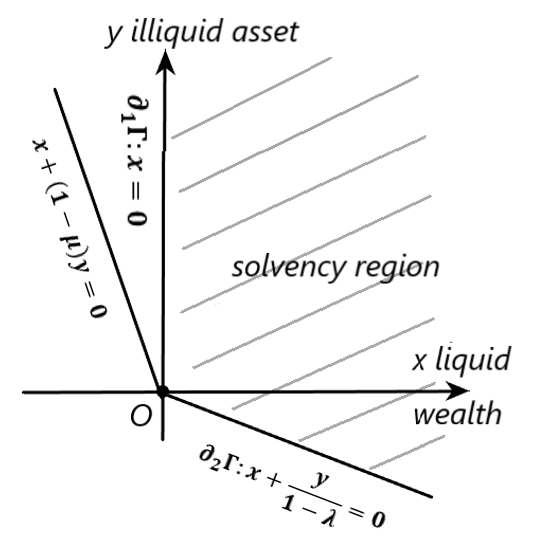} 
	\caption{The solvency region}
	\label{figure1}
\end{figure}

For an initial position $\left(X_{0-},Y_{0-}\right)=\left(x_0,y_0\right)\in\overline{\Lambda}$, we denote by $\mathcal{A}\left(x_0,y_0\right)$ the set of admissible policy $\left(\pi,c,L,M\right)$, under which $\left(X,Y\right)$ given by Eqs.~(\ref{xt})-(\ref{yt}) is always in $\overline{\Gamma}$. Then the objective of the agent is to maximize the expected lifetime discounted utility from consumption and liquidity preference, i.e.,
\begin{eqnarray}\label{problem}
	v\left(x,y\right)=\sup_{\left(\pi,c,L,M\right)\in\mathcal{A}\left(x,y\right)}\mathbb{E}
	\left[\int_{0}^{\infty}e^{-\beta t}U\left(c_t,X_t\right)dt\right].
\end{eqnarray}

Before investigating Problem \eqref{problem}, we propose some assumptions on the parameters. 
In \cite{Davis1990},  the necessary and sufficient condition for finiteness of  the value function in the optimization problem without trading the illiquid asset to be finite is
\begin{eqnarray}
	\beta-rp>\frac{p\alpha_1^2}{2\left(1-p\right)\sigma_1^2}.\label{assum1}
\end{eqnarray}
\cite{hobson2019multi} formulates another assumption
\begin{eqnarray}
	\beta-rp>\frac{p\left(\alpha_1^2\sigma_2^2+\alpha_2^2\sigma_1^2-2\rho\alpha_1\alpha_2\sigma_1\sigma_2\right)}{2\left(1-\rho^2\right)\left(1-p\right)\sigma_1^2\sigma_2^2},
	\label{assum2}
\end{eqnarray}
which is necessary and sufficient for the value function in the optimization problem without transaction costs to be finite. In the above two problems, there is no illiquid asset, and then liquidity preference is no longer applicable, so that $\theta$ can only be taken as $1$.

In what follows, we will assume the following.
\begin{assumption}\label{ass}
	Throughout this paper, we assume that (\ref{assum1}) and (\ref{assum2}) hold.
\end{assumption}

\section{\bf Properties of the value function}
In this  section, similar to \cite{Sherve1994}, we present several properties about the value function in Problem \eqref{problem}.
\subsection{\bf The HJB equation}
Given $\left(x_0,y_0\right)\in\overline{\Lambda},\left(\pi,c,L,M\right)\in\mathcal{A}\left(x_0,y_0\right),\varphi\in\mathcal{C}^2\left(\Lambda\right)$ and an almost surely finite stopping time $\tau$, It\^{o}'s rule yields
\begin{eqnarray}
	\varphi\left(x_0,y_0\right)\!\!&=&\!\!e^{-\beta\tau}\varphi\left(X_{\tau},Y_{\tau}\right)+\!\!\int_{0}^{\tau}e^{-\beta t}\!\!\left(\mathcal{L}\varphi+c_t\varphi_x\right) dt- \!\!\int_{0}^{\tau}e^{-\beta t}\!\!\sigma_1\pi_tX_t\varphi_x dB^1_t -\!\!\int_{0}^{\tau}e^{-\beta t}\!\!\sigma_2Y_t\varphi_y dB^2_t\nonumber\\
	&+&\int_{0}^{\tau}e^{-\beta t}\left\{\left[-\left(1-\mu\right)\varphi_x+\varphi_y\right]dM^c_t+\left[\varphi_x-\left(1-\lambda\right)\varphi_y\right]dL^c_t\right\}\nonumber\\
	&+&\sum_{0\le t\le \tau} e^{-\beta t}\left[\varphi\left(X_{t-},Y_{t-}\right)-\varphi\left(X_t,Y_t\right)\right],\label{ito}
\end{eqnarray}
where the second-order differential operator $\mathcal{L}$ is defined by 
$$\left(\mathcal{L}\varphi\right)\left(x,y\right)=\beta\varphi-\left(r+\alpha_1\pi\right)x\varphi_x-\left(r+\alpha_2\right)y\varphi_y-\frac{1}{2}\sigma_1^2\pi^2x^2\varphi_{xx}-\frac{1}{2}\sigma_2^2y^2\varphi_{yy}-\rho\sigma_1\sigma_2\pi xy\varphi_{xy}.$$

Then the Hamilton-Jacobi-Bellman (HJB) equation of Problem \eqref{problem} is
\begin{eqnarray}\label{zyshjb}
	\min\left\{\min_{\pi,c}\{\mathcal{L}\varphi+c\varphi_x-U\left(c,x\right)\},-\left(1-\mu\right)\varphi_x+\varphi_y,\varphi_x-\left(1-\lambda\right)\varphi_y\right\}=0,
\end{eqnarray}
which is equivalent to
\begin{eqnarray}
	\min\left\{\min_{\pi}\{\mathcal{L}\varphi-\tilde{U}\left(x,\varphi_x\right)\},-\left(1-\mu\right)\varphi_x+\varphi_y,\varphi_x-\left(1-\lambda\right)\varphi_y\right\}=0,\label{hjb}
\end{eqnarray}
where $$\tilde{U}\left(x,\tilde{x}\right)=\frac{1-\theta p}{p}\theta^{\frac{\theta p}{1-\theta p}}x^{\frac{\left(1-\theta\right)p}{1-\theta p}}{\tilde{x}}^{-\frac{\theta p}{1-\theta p}}$$ 
and the minimum point is attained at
\begin{eqnarray}
	c_*=\theta^{\frac{1}{1-\theta p}}x^{\frac{\left(1-\theta\right)p}{1-\theta p}}\varphi_x^{-\frac{1}{1-\theta p}}.\nonumber
\end{eqnarray}

%% For convenience, the left-hand side of the HJB equation is denoted as $F\left(\left(x,y\right),\varphi,\left(\varphi_x,\varphi_y\right),\left(\begin{matrix}
%%	\varphi_{xx} & \varphi_{xy} \\ \varphi_{xy} & \varphi_{yy}
%% \end{matrix}\right)\right)$.

Furthermore, if we assume $\varphi_{xx} <0$, which is naturally satisfied when $\varphi$ is strictly concave, $\mathcal{L}\varphi-\tilde{U}\left(x,\varphi_x\right)$ takes the minimum value at the point $$\pi_*=-\frac{\alpha_1\varphi_x+\rho\sigma_1\sigma_2y\varphi_{xy}}{\sigma_1^2x\varphi_{xx}}. $$ Then we define the operator
\begin{eqnarray}
	\overline{\mathcal{L}}\varphi&:=&\mathcal{L}\varphi-\tilde{U}\left(x,\varphi_x\right)|_{\pi=\pi_*}\nonumber\\
	&=&\beta\varphi-rx\varphi_x-\left(r+\alpha_2\right)y\varphi_y-\frac{1}{2}\sigma_2^2y^2\varphi_{yy} +\frac{\left(\alpha_1\varphi_x+\rho\sigma_1\sigma_2y\varphi_{xy}\right)^2}{2\sigma_1^2\varphi_{xx}}-\tilde{U}\left(x,\varphi_x\right).\label{fxxhjb}
\end{eqnarray}

The original HJB equation (\ref{zyshjb}) is linear, as is the HJB equation (\ref{hjb}) with respect to the second-order term. However, the final operator (\ref{fxxhjb}) and the corresponding HJB equation are nonlinear.

\subsection{\bf Properties of the value function}

This subsection examines several key properties of value functions, including homotheticity, convexity, boundedness, and continuity. Analyzing these properties enables us to understand the shape and structure of the optimal policy, which we explore in the following section.

\begin{proposition}\label{propo}
	The value function $v$ defined in Problem \eqref{problem} has the following properties:
	\item[(1)] $v$ has the homotheticity property 
	$$ v\left(mx,my\right)=m^pv\left(x,y\right), \quad\forall \left(x,y\right)\in\overline{\Lambda},\quad m>0.$$
	\item[(2)] $v$ is concave in $\overline{\Lambda}$.
	\item[(3)] $v$ is continuous in $\Lambda$. In particular, $v$ is continuous in $\Gamma\cup\partial_1\Gamma$.
	\item[(4)] $v$ has the lower bound in $\overline{\Lambda}$
	\begin{eqnarray}
		v\left(x,y\right)\ge\begin{cases}
			\frac{C_*}{p}\left(x+\left(1-\mu\right)y\right)^p, &\forall \left(x,y\right)\in\overline{\Lambda}, y\ge0,\\
			\frac{C_*}{p}\left(x+\frac{y}{1-\lambda}\right)^p, &\forall \left(x,y\right)\in\overline{\Lambda}, y<0,
		\end{cases}
	\end{eqnarray}
	where 
	\begin{eqnarray}
		C_*=\left(1-\theta p\right)^{1-\theta p}\theta^{\theta p}\left[\beta-rp-\frac{p\alpha_1^2}{2\left(1-p\right)\sigma_1^2}\right]^{\theta p-1}.\label{c*}
	\end{eqnarray}
	\item[(5)] $v$ satisfies the principle of dynamic programming. Let $\left(x,y\right)\in\Lambda$, $\Omega$ be an open subset of $\Lambda$ containing $\left(x,y\right)$, and $$\tau :=\inf\{t\ge0\mid \left(X_t,Y_t\right)\notin\overline{\Omega}\}.$$ Then, $\forall t\in\left[0,\infty\right]$, we have
	\begin{eqnarray}
		v\left(x,y\right)=\sup_{\left(\pi,c,L,M\right)\in\mathcal{A}\left(x,y\right)}\mathbb{E}\left[\int_{0}^{t\wedge\tau}e^{-\beta s}u\left(c_s,X_s\right)ds+1_{\{t\wedge\tau<\infty\}}e^{-\beta\left(t\wedge\tau\right)}v\left(X_{t\wedge\tau},Y_{t\wedge\tau}\right)\right].
	\end{eqnarray}
\end{proposition}
\begin{proof}
	We prove the properties one by one.
	Property (1) is trivial and follows from the fact $$\left(\pi,c,L,M\right)\in\mathcal{A}\left(x,y\right) \quad \iff \quad \left(\pi,mc,mL,mM\right)\in\mathcal{A}\left(mx,my\right).$$
	
	Using the Jacobi matrix of $U$, we can conclude that $U$ is concave with respect to $(c, x)$. Following a similar argument to Proposition 3.1 in \cite{Sherve1994}, we can show that $v$ is also concave, which verifies Property (2). Because a concave function is continuous on the interior of its domain, we can directly deduce Property (3) from Property (2).
	
	Regarding Property (4), one admissible policy is to immediately liquidate the illiquid asset and jump to the $x$-axis, after which there is no further transfer between the illiquid asset and liquid wealth. The resulting optimization problem is equivalent to the problem of optimizing without trading the illiquid asset, and the value function $\tilde{v}$ can be computed using the method described in \cite{Davis1990}. To be specific, the HJB equation for $\tilde{v}$ is
	$$\max_{c,\pi}\left\{-\beta\tilde{v}+\left(r+\alpha_1\pi\right)x\tilde{v}_x+\frac{1}{2}\sigma_1^2\pi^2x^2\tilde{v}_{xx}-c\tilde{v}_x+U\left(c,x\right)\right\}=0,$$ that is,
	$$-\beta\tilde{v}-rx\tilde{v}_x+\frac{\alpha_1^2\tilde{v}_x^2}{\sigma_1^2\tilde{v}_{xx}}+\frac{1-\theta p}{p}\theta^{\frac{\theta p}{1-\theta p}}x^{\frac{\left(1-\theta\right)p}{1-\theta p}}\tilde{v}_x^{\frac{\theta p}{\theta p-1}}=0,$$
	and the value function $\tilde{v}$ is $$\tilde{v}\left(x\right)=\frac{C_*}{p}x^p,$$ where $C_*$ is given by Eq.~(\ref{c*}). 
	
	To prove Property (5), we can refer to Corollary 4.2 in \cite{Sherve1994} by replacing  $U_p\left(x\right)$ with $U\left(c,x\right)$.
\end{proof}

In order to prove the continuity of $v$ on $\partial_2\Gamma$, we only need to derive the value of $v$ on $\partial_2\Gamma$. According to the following proposition, we know $v=0$ when $0<p<1$ and $v=-\infty$ when $p<0$ on $\partial_2\Gamma$.

\begin{proposition}\label{propo2}
	If $\left(x_0,y_0\right)\in\partial_2\Gamma$, all admissible policies must involve using the liquid wealth to cover the short position in the illiquid asset and remaining at position $(0,0)$.
\end{proposition}
\begin{proof}
	For any nonnegative $(L_0,M_0)$ at time zero. Note that $X_0=x_0-L_0+\left(1-\mu\right)M_0$ and $Y_0=y_0+\left(1-\lambda\right)L_0-M_0$, we obtain $$X_0+\frac{Y_0}{1-\lambda}=\left(1-\mu-\frac{1}{1-\lambda}\right)M_0\ge0,$$ which implies $M_0=0$ and then $\left(X_0,Y_0\right)\in\partial_2\Gamma$. 
	
	Let $$\tau_n:=1\wedge\inf\{t\ge0\mid|\pi_tX_t|>n\; \mbox{or} Y_t\notin\left(y_0-1,0\right)\}$$ and $$\tau:=1\wedge\inf\{t\ge0\mid Y_t\notin\left(y_0-1,0\right)\},$$ then
	\begin{eqnarray}
		0&\le& e^{-r\tau_n } \left(X_{\tau_n}+\frac{1}{1-\lambda}Y_{\tau_n}\right)\nonumber\\
		&=& \int_{0}^{\tau_n}e^{-rt}\left[\alpha_1\pi_tX_t+\frac{\alpha_2}{1-\lambda}Y_t-c_t\right]dt+\int_{0}^{\tau_n}e^{-rt}\left(1-\mu-\frac{1}{1-\lambda}\right)dM_t\nonumber\\
		&+&\int_{0}^{\tau_n}e^{-rt}\sigma_1\pi_tX_tdB^1_t +\int_{0}^{\tau_n}e^{-rt}\frac{\sigma_2}{1-\lambda}Y_tdB^2_t\nonumber\\
		&\le&\int_{0}^{\tau_n}e^{-rt}\pi_tX_td\left(\alpha_1t+\sigma_1B^1_t\right) +\int_{0}^{\tau_n}e^{-rt}\frac{\sigma_2}{1-\lambda}Y_tdB^2_t, \qquad \mathbb{P}-a.s..\nonumber
	\end{eqnarray}
	It is worth noting that $\frac{B^1-\rho B^2}{\sqrt{1-\rho^2}}$ and $B^2$ are two standard Brownian motions that are independent of each other. Applying Girsanov's Theorem, we can find an equivalent measure $\mathbb{Q}$ on which $B^2$ and $\left\{\alpha_1t+\sigma_1\left(B^1_t-\rho B^2_t\right)\right\}_{t\ge 0}$ are two independent Brownian motions. As a result, we have
	\begin{eqnarray}
		\mathbb{E}^{\mathbb{Q}} \left[\int_{0}^{\tau_n}e^{-rt}\pi_tX_td\left(\alpha_1t+\sigma_1B^1_t\right) + e^{-rt}\frac{\sigma_2}{1-\lambda}Y_tdB^2_t\right]=0,\nonumber
	\end{eqnarray}
	from which we have
	\begin{eqnarray}
		\int_{0}^{\tau_n}e^{-rt}\pi_tX_td\left(\alpha_1t+\sigma_1B^1_t\right) + e^{-rt}\frac{\sigma_2}{1-\lambda}Y_tdB^2_t=0,\quad \mathbb{Q}-a.s.,\nonumber
	\end{eqnarray}
	implying $$\tau_n=0,\quad \mathbb{Q}-a.s.,$$ and then $$\tau_n=0,\quad \mathbb{P}-a.s..$$ 
	Letting $n\rightarrow +\infty$, we have $$\tau=0,\quad \mathbb{P}-a.s.,$$ because $\tau_n\uparrow\tau$. Then we must have $Y_0\ge0$ as $M_0=0$ implies $Y_0\ge y_0$, but we have proved $\left(X_0,Y_0\right)\in\partial_2\Gamma$, we conclude $\left(X_0,Y_0\right)=\left(0,0\right)$.
\end{proof}

To establish continuity of $v$ on $\partial_2\Gamma$, we need to discuss two cases separately, $0<p<1$ and $p<0$. When $0<p<1$, we can derive the results regarding the form of $v$ near the boundaries $\partial_1\Gamma$ and $\partial_2\Gamma$. These results are also applicable in the subsequent section.

\begin{theorem}\label{v=phi}
	Assume $0<p<1$, there exist $\delta_1>0$ sufficiently small and $\delta_2<1$ sufficiently close to $1$ such that
	\begin{eqnarray}
		v\left(x,y\right)&=&v\left(\frac{\delta_1\left(x+\left(1-\mu\right)y\right)}{1-\mu+\delta_1},\frac{x+\left(1-\mu\right)y}{1-\mu+\delta_1}\right)\nonumber\\
		&=&\frac{A}{p}\left(x+\left(1-\mu\right)y\right)^p \qquad\qquad\qquad\qquad\qquad if \quad x\ge0,\;y>0,\;\frac{x}{y}<\delta_1,\nonumber\\
		v\left(x,y\right)&=&v\left(\frac{x+\frac{y}{1-\lambda}}{1-\delta_2},-\frac{\delta_2\left(1-\lambda\right)\left(x+\frac{y}{1-\lambda}\right)}{1-\delta_2}\right)\nonumber\\
		&=&\frac{B}{p}\left(x+\frac{y}{1-\lambda}\right)^p \qquad\qquad\qquad\qquad\qquad if \quad -\delta_2<\frac{y}{\left(1-\lambda\right)x}\le-1\nonumber,
	\end{eqnarray}
	where $$A:=\frac{pv\left(\delta_1,1\right)}{\left(1-\mu+\delta_1\right)^p}, \quad\quad B:=\frac{pv\left(1,-\left(1-\lambda\right)\delta_2\right)}{\left(1-\delta_2\right)^p}.$$
\end{theorem}
\begin{proof}
	We first show the results of $v$ near $\partial_1\Gamma$. Define the wedge $$D_1:=\{\left(x,y\right)\in\Lambda\mid  x\ge0,y>0,\frac{x}{y}<\delta_1\}$$ and a strictly concave function
	$$\varphi\left(x,y\right):=\frac{A}{p}\left(x+\left(1-\mu\right)y\right)^p.$$ From Property (4) of Proposition \ref{propo}, we have
	\begin{eqnarray}
		\frac{A}{p}=\frac{v\left(\delta_1,1\right)}{\left(1-\mu+\delta_1\right)^p}\ge \frac{C_*}{P}. \label{a>c}
	\end{eqnarray}  
	
	Note that $$\varphi\left(x,y\right)=v\left(\frac{\delta_1\left(x+\left(1-\mu\right)y\right)}{1-\mu+\delta_1},\frac{x+\left(1-\mu\right)y}{1-\mu+\delta_1}\right),$$ 
	we have $\varphi\le v$ on  $D_1$ and $\varphi\ge v$ on $\Lambda\setminus D_1$. Moreover, for $\left(x,y\right)\in D_1$,  $$-\left(1-\mu\right)\varphi_x+\varphi_y\ge0,\varphi_x-\left(1-\lambda\right)\varphi_y\ge0.$$ Using Eq.~(\ref{a>c}) and Assumption (\ref{assum2}), we obtain
	\begin{eqnarray}
		\!\!\overline{\mathcal{L}}\varphi\!\!\!\!\!\!&=&\!\!\!\! A\left(x+\left(1-\mu\right)y\right)^p\times\nonumber\\
		&&\!\!\!\!\left\{\frac{\beta-rp}{p}-\frac{1-p}{2}\left[\frac{\alpha_1}{\left(1-p\right)\sigma_1}\left(1+\left(1-\mu\right)\frac{y}{x}\right)-\rho\sigma_2\left(1-\mu\right)\frac{y}{x}\right]^2\frac{x^2}{\left(x+\left(1-\mu\right)y\right)^2}\right.\nonumber\\
		&&\!\!\!\!-\alpha_2\frac{\left(1-\mu\right)y}{x+\left(1-\mu\right)y}+\frac{1}{2}\sigma_2^2\left(1-p\right)\left[\frac{\left(1-\mu\right)y}{x+\left(1-\mu\right)y}\right]^2\nonumber\\
		&&\!\!\!\!\left.-\frac{1-\theta p}{p}\theta^{\frac{\theta p}{1-\theta p}}A^{-\frac{1}{1-\theta p}}\left(\frac{x}{x+\left(1-\mu\right)y}\right)^{\frac{\left(1-\theta\right)p}{1-\theta p}}\right\}\nonumber\\
		&\ge&\!\!\!\! A\left(x+\left(1-\mu\right)y\right)^p\times\nonumber\\
		&&\!\!\!\!\left\{\frac{\beta-rp}{p}-\frac{1-p}{2}\left[\frac{\alpha_1}{\left(1-p\right)\sigma_1}\left(1+\left(1-\mu\right)\frac{y}{x}\right)-\rho\sigma_2\left(1-\mu\right)\frac{y}{x}\right]^2\frac{x^2}{\left(x+\left(1-\mu\right)y\right)^2}\right.\nonumber\\
		&&\!\!\!\!-\alpha_2\frac{\left(1-\mu\right)y}{x+\left(1-\mu\right)y}+\frac{1}{2}\sigma_2^2\left(1-p\right)\left[\frac{\left(1-\mu\right)y}{x+\left(1-\mu\right)y}\right]^2\nonumber\\
		&&\!\!\!\!\left.-\left[\frac{\beta-rp}{p}-\frac{\alpha_1^2}{2\left(1-p\right)\sigma_1^2}\right]\left(\frac{x}{x+\left(1-\mu\right)y}\right)^{\frac{\left(1-\theta\right)p}{1-\theta p}}\right\}\nonumber\\
		&\overset{t:=\frac{\left(1-\mu\right)y}{x}}{=}&\!\!\!\! A\left(x+\left(1-\mu\right)y\right)^p\left\{ \left[\frac{\beta-rp}{p}-\frac{\alpha_1^2}{2\left(1-p\right)\sigma_1^2}\right]\left[1-\left(1+t\right)^{-\frac{\left(1-\theta\right)p}{1-\theta p}}\right]\right.\nonumber\\
		&&\!\!\!\!+\left.\frac{1}{2}\left(1-\rho^2\right)\left(1-p\right)\sigma_2^2\left(\frac{t}{1+t}\right)^2+\left(\frac{\rho\alpha_1\sigma_2}{\sigma_1}-\alpha_2\right)\frac{t}{1+t}\right\}\nonumber\\
		&\ge&\!\!\!\! A\left(x+\left(1-\mu\right)y\right)^p\left\{ \left[\frac{\beta-rp}{p}-\frac{\alpha_1^2}{2\left(1-p\right)\sigma_1^2}\right]\left[1-\left(1+t\right)^{-\frac{\left(1-\theta\right)p}{1-\theta p}}\right]\right.\nonumber\\
		&&\!\!\!\!-\left.\frac{\left(\rho\alpha_1\sigma_2-\alpha_2\sigma_1\right)^2}{2\left(1-\rho^2\right)\left(1-p\right)\sigma_1^2\sigma_2^2}\right\}\nonumber\\
		&:=&\!\!\!\! A\left(x+\left(1-\mu\right)y\right)^p f\left(t\right).\label{ft}
	\end{eqnarray}
	As $t\uparrow+\infty$, from assumption (\ref{assum2}), we have $$f\left(t\right)\longrightarrow \frac{\beta-rp}{p}-\frac{\alpha_1^2\sigma_2^2+\alpha_2^2\sigma_1^2-2\rho\alpha_1\alpha_2\sigma_1\sigma_2}{2\left(1-\rho^2\right)\left(1-p\right)\sigma_1^2\sigma_2^2}>0.$$ As such, there exists a sufficiently small $\delta_1>0$ such that $\overline{\mathcal{L}}\varphi\ge 0$ on $D_1$.
	
	Define the trapezoid $$H_n:=\{\left(x,y\right)\in\overline{D}_1\mid\frac{1}{n}\le y\le n\}$$ 
	and the stopping times
	$$ T_n :=\inf\{t\ge0\mid\left(X_t,Y_t\right)\notin H_n\},\; T :=\inf\{t\ge0\mid\left(X_t,Y_t\right)\notin \overline{D}_1\},\;\kappa_m:=\inf\{t\ge0\mid|\pi_t|>m\}.$$ Then $\lim\limits_{n\rightarrow\infty} T_n=T$ almost surely and $$\{T<\infty\}=\bigcup_{n=1}^{\infty}\{T_n=T\le n\}.$$
	
	For any $(x,y)\in D_1$ and some policy $(\pi,c,L,M)\in\mathcal{A}(x,y)$, note that $\pi$ and $(X_t,Y_t)$ are bounded for $t<T_n\wedge\kappa_m$. Using the formula in Eq.~(\ref{ito}) and taking $\tau=T_n\wedge\kappa_m\wedge n$, we obtain
	\begin{eqnarray}
		\varphi\left(x,y\right)&=&\mathbb{E}\left[e^{-\beta T_n\wedge \kappa_m\wedge n}\varphi\left(X_{T_n\wedge \kappa_m\wedge n},Y_{T_n\wedge \kappa_m\wedge n}\right)\right]+\mathbb{E}\left[\int_{0}^{T_n\wedge \kappa_m\wedge n}e^{-\beta t}\left(\mathcal{L}\varphi+c_t\varphi_x\right) dt\right]\nonumber\\
		&&+\mathbb{E}\left[\int_{0}^{T_n\wedge \kappa_m\wedge n}e^{-\beta t}\left\{\left[-\left(1-\mu\right)\varphi_x+\varphi_y\right]dM^c_t+\left[\varphi_x-\left(1-\lambda\right)\varphi_y\right]dL^c_t\right\}\right]\nonumber\\
		&&+\sum_{0\le t\le T_n\wedge \kappa_m\wedge n} e^{-\beta t}\mathbb{E}\left[\varphi\left(X_{t-},Y_{t-}\right)-\varphi\left(X_t,Y_t\right)\right]\nonumber\\
		&\ge&\mathbb{E}\left[e^{-\beta T_n\wedge \kappa_m\wedge n}\varphi\left(X_{T_n\wedge \kappa_m\wedge n},Y_{T_n\wedge \kappa_m\wedge n}\right)\right]+\mathbb{E}\left[\int_{0}^{T_n\wedge \kappa_m\wedge n}e^{-\beta t}U\left(c_t,X_t\right) dt\right].\nonumber		
	\end{eqnarray}
    Taking the limit $m\rightarrow \infty$,
	\begin{eqnarray}
		\varphi\left(x,y\right)&\ge&
		\mathbb{E}\left[e^{-\beta T_n\wedge n}\varphi\left(X_{T_n\wedge n},Y_{T_n\wedge n}\right)\right]+\mathbb{E}\left[\int_{0}^{T_n\wedge n}e^{-\beta t}U\left(c_t,X_t\right) dt\right],\nonumber\\
		&\ge&\mathbb{E}\left[1_{\{T_n=T\le n\}}e^{-\beta T}v\left(X_{T},Y_{T}\right)\right]+\mathbb{E}\left[\int_{0}^{T_n\wedge n}e^{-\beta t}U\left(c_t,X_t\right) dt\right],\nonumber
	\end{eqnarray}
    and then taking the limit $n\rightarrow \infty$,
	\begin{eqnarray}
		 \varphi\left(x,y\right)&\ge&\mathbb{E}\left[1_{\{T<\infty\}}e^{-\beta T}v\left(X_{T},Y_{T}\right)\right]+\mathbb{E}\left[\int_{0}^{T}e^{-\beta t}U\left(c_t,X_t\right) dt\right].\nonumber
	\end{eqnarray}
	Maximizing the right side over $(\pi,c,L,M)\in\mathcal{A}(x,y)$ and applying the principle of dynamic programming (Property (5) of Proposition \ref{propo}), we conclude $\varphi(x,y)\ge v(x,y)$ on $D_1$. Hence, we have $\varphi(x,y)=v(x,y)$ on $D_1$.
	
	The proof of the second part (the form of $v$ near $\partial_2\Gamma$) is similar. Define the wedge $$D_2:=\{\left(x,y\right)\in\Lambda\mid -\delta_2<\frac{y}{\left(1-\lambda\right)x}\le-1\}$$ and a strictly concave function
	$$\varphi\left(x,y\right):=\frac{B}{p}\left(x+\frac{y}{1-\lambda}\right)^p.$$ From the property (4) of Proposition \ref{propo}, we have
	\begin{eqnarray}
		\frac{B}{p}=\frac{v\left(1,-\left(1-\lambda\right)\delta_2\right)}{\left(1-\delta_2\right)^p}\ge \frac{C_*}{P}. \label{b>c}
	\end{eqnarray}  
	As $$\varphi\left(x,y\right)=v\left(\frac{x+\frac{y}{1-\lambda}}{1-\delta_2},-\frac{\delta_2\left(1-\lambda\right)\left(x+\frac{y}{1-\lambda}\right)}{1-\delta_2}\right),$$ we have $\varphi\le v$ on $D_2$ and $\varphi\ge v$ on $\Lambda\setminus D_2$. Moreover, for $(x,y)\in D_2$, we have $$-\left(1-\mu\right)\varphi_x+\varphi_y\ge0,\quad \varphi_x-\left(1-\lambda\right)\varphi_y\ge0.$$ Using Eq.~(\ref{a>c}) and Assumption (\ref{assum2}), we obtain
	\begin{eqnarray}
		\!\!\overline{\mathcal{L}}\varphi\!\!\!\!\!\!&=&\!\!\!\! B\left(x+\frac{y}{1-\lambda}\right)^p\times\nonumber\\
		&&\!\!\!\!\left\{\frac{\beta-rp}{p}-\frac{1-p}{2}\left[\frac{\alpha_1}{\left(1-p\right)\sigma_1}\left(1+\frac{y}{\left(1-\lambda\right)x}\right)-\rho\sigma_2\frac{y}{\left(1-\lambda\right)x}\right]^2\frac{\left(1-\lambda\right)^2x^2}{\left(\left(1-\lambda\right)x+y\right)^2}\right.\nonumber\\
		&&\!\!\!\!-\alpha_2\frac{y}{\left(1-\lambda\right)x+y}+\frac{1}{2}\sigma_2^2\left(1-p\right)\left[\frac{y}{\left(1-\lambda\right)x+y}\right]^2\nonumber\\
		&&\!\!\!\!\left.-\frac{1-\theta p}{p}\theta^{\frac{\theta p}{1-\theta p}}B^{-\frac{1}{1-\theta p}}\left(\frac{\left(1-\lambda\right)x}{\left(1-\lambda\right)x+y}\right)^{\frac{\left(1-\theta\right)p}{1-\theta p}}\right\}\nonumber\\
		&\ge&\!\!\!\! B\left(x+\frac{y}{1-\lambda}\right)^p\times\nonumber\\
		&&\!\!\!\!\left\{\frac{\beta-rp}{p}-\frac{1-p}{2}\left[\frac{\alpha_1}{\left(1-p\right)\sigma_1}\left(1+\frac{y}{\left(1-\lambda\right)x}\right)-\rho\sigma_2\frac{y}{\left(1-\lambda\right)x}\right]^2\frac{\left(1-\lambda\right)^2x^2}{\left(\left(1-\lambda\right)x+y\right)^2}\right.\nonumber\\
		&&\!\!\!\!-\alpha_2\frac{y}{\left(1-\lambda\right)x+y}+\frac{1}{2}\sigma_2^2\left(1-p\right)\left[\frac{y}{\left(1-\lambda\right)x+y}\right]^2\nonumber\\
		&&\!\!\!\!\left.-\left[\frac{\beta-rp}{p}-\frac{\alpha_1^2}{2\left(1-p\right)\sigma_1^2}\right]\left(\frac{\left(1-\lambda\right)x}{\left(1-\lambda\right)x+y}\right)^{\frac{\left(1-\theta\right)p}{1-\theta p}}\right\}\nonumber\\
		&\overset{t:=\frac{y}{\left(1-\lambda\right)x}}{=}&\!\!\!\! B\left(x+\frac{y}{1-\lambda}\right)^p\frac{1}{\left(1+t\right)^2}\left\{ \left[\frac{\beta-rp}{p}-\frac{\alpha_1^2}{2\left(1-p\right)\sigma_1^2}\right]\left[\left(1+t\right)^2-\left(1+t\right)^{2-\frac{\left(1-\theta\right)p}{1-\theta p}}\right]\right.\nonumber\\
		&&\!\!\!\!+\left.\frac{1}{2}\left(1-\rho^2\right)\left(1-p\right)\sigma_2^2t^2+\left(\frac{\rho\alpha_1\sigma_2}{\sigma_1}-\alpha_2\right)t\left(1+t\right)\right\}\nonumber\\
		&:=&\!\!\!\! B\left(x+\frac{y}{1-\lambda}\right)^p\frac{1}{\left(1+t\right)^2}g\left(t\right)\label{gt}.
	\end{eqnarray}
	
	As $t\downarrow -1$, we have $$g\left(t\right)\longrightarrow \frac{1}{2}\left(1-\rho^2\right)\left(1-p\right)\sigma_2^2 > 0.$$ Therefore, for $\delta_2<1$ and sufficiently close to $1$, we have $\overline{\mathcal{L}}\varphi\ge 0$ on $D_2$.
	
	Define the trapezoids and stopping times
	$$K_n:=\{\left(x,y\right)\in\overline{D_2}\mid\frac{1}{n}\le x\le n\}, S_n :=\inf\{t\ge0\mid\left(X_t,Y_t\right)\notin K_n\}, S :=\inf\{t\ge0\mid\left(X_t,Y_t\right)\notin \overline{D}_2\}.$$ Then $\lim\limits_{n\rightarrow\infty} S_n=S$ almost surely and $$\{S<\infty\}=\bigcup_{n=1}^{\infty}\{S_n=S\le n\}.$$ 
	We can still prove $\varphi(x,y)\ge v(x,y)$ on $D_2$, and therefore, $\varphi(x,y)= v(x,y)$ on $D_2$.
\end{proof}
\begin{remark}
	The condition $\theta\neq 1$ is necessary for the existence of $\delta_1$. If $\theta=1$, then $f\left(t\right)$ in Eq.~(\ref{ft}) becomes $$f\left(t\right)=-\frac{\left(\rho\alpha_1\sigma_2-\alpha_2\sigma_1\right)^2}{2\left(1-\rho^2\right)\left(1-p\right)\sigma_1^2\sigma_2^2},$$ and the subsequent proof is no longer valid. When $\theta=1$, the existence of $\delta_1$ depends on other parameters, which will be demonstrated in the numerical analysis section. In this paper, we assume $\theta\neq1$ and only consider the case of $\theta=1$ in the numerical analysis section.
\end{remark}

By referring to Theorem \ref{v=phi}, we observe $v=0$ on $\partial_2\Gamma$ when $0<p<1$. Therefore, we can directly deduce the following:
\begin{corollary}\label{feikong}
	Assuming $0<p<1$, $v$ is continuous on $\partial_2\Gamma$.
\end{corollary}

We will now demonstrate the continuity of $v$ on $\partial_2\Gamma$ in the case where $p<0$. As $v=-\infty$ on $\partial_2\Gamma$, we need to prove that $v$ has a limit of $-\infty$ at $\partial_2\Gamma$, which can be stated in the following proposition:

\begin{proposition}\label{propo4}
	Assuming $p<0$, we can conclude that $v$ has a limit of $-\infty$ on  $\partial_2\Gamma$. Consequently, we can also deduce that $v$ is continuous on $\partial_2\Gamma$.
\end{proposition}
\begin{proof}
	Let us consider $0<p<1$ as a variable. We can define the increasing function $$G(p):=\frac{\beta-rp}{p}-\frac{\alpha_1^2}{2\left(1-p\right)\sigma_1^2},$$ which is positive for small $p>0$. Define $$\varphi:=\frac{D\left(p\right)}{p}\left(x+\frac{y}{1-\lambda}\right)^p,$$ where $$D\left(p\right) := \theta^{\theta p}\left[\frac{pG\left(p\right)}{1-\theta p}\right]^{\theta p-1}, \quad i.e.,\quad G\left(p\right)=-\frac{1-\theta p}{p}\theta^{\frac{\theta p}{1-\theta p}}D\left(p\right)^{-\frac{1}{1-\theta p}}.$$ Let $t:=\frac{y}{\left(1-\lambda\right)x}$, then 
	\begin{eqnarray}
		\overline{\mathcal{L}}\varphi &=& D\left(p\right)\left(x+\frac{y}{1-\lambda}\right)^p\frac{1}{\left(1+t\right)^2}g\left(t\right),\nonumber
	\end{eqnarray}
	where the function $g\left(t\right)$ is defined in Eq.~(\ref{ft}). From the proof of Theorem \ref{v=phi}, we can establish $v\leq \varphi$ on $D_2(p)$. Additionally, as $f(t)$ is monotonically increasing with respect to $p$, we can choose $D_2(p)$ decreasing with $p$.
	
	For $p<0$, let $0<q<\epsilon<1$ be chosen. If $(x,y)\in D_2(\epsilon)$, the inequality $$\frac{c^p}{p}\leq \log c -\frac{1}{p} \leq \frac{c^q}{q}-\frac{1}{q}-\frac{1}{p}$$ implies $$v_p\leq v_q-\frac{1}{\beta q} -\frac{1}{\beta p}\leq\frac{1}{q}\left[D(q)\left(x+\frac{y}{1-\lambda}\right)^q-\frac{1}{\beta}\right]-\frac{1}{\beta p}.$$
	As $(x,y)\longrightarrow (x_0,y_0)\in\partial_2\Gamma$, we can take $q\downarrow 0$ such that
	\begin{eqnarray}
		\lim\limits_{\left(x,y\right)\rightarrow\left(x_0,y_0\right)}v_p\left(x,y\right)&\le&\lim\limits_{\left(x,y\right)\rightarrow\left(x_0,y_0\right)}\lim\limits_{q\downarrow 0}\frac{1}{q}\left[D\left(q\right)\left(x+\frac{y}{1-\lambda}\right)^q-\frac{1}{\beta}\right]-\frac{1}{\beta p}\nonumber\\
		&=&\lim\limits_{\left(x,y\right)\rightarrow\left(x_0,y_0\right)}\frac{1}{\beta}\left[\log\left(x+\frac{y}{1-\lambda}\right)+\theta\log\beta+\frac{r-\beta}{\beta}+\frac{\alpha_1^2}{2\sigma_1^2\beta}\right]-\frac{1}{\beta p}\nonumber\\
		&=&-\infty.\nonumber
	\end{eqnarray}
	Then $v$ has a limit $-\infty$ on $\partial_2\Gamma$ and the proposition holds. 
\end{proof}

Specifically, if $p<0$,  then $v<0$ near $\partial_2\Gamma$ and we can obtain a stronger conclusion.
\begin{corollary}
	Assume $p<0$, then $v<0$ on $\overline{\Gamma}$.
\end{corollary}
\begin{proof}
	$v$ is negative at some point, then homotheticity implies that $v$ is negative on a ray passing through $(0,0)$. From this ray, we can reach every point in $\overline{\Gamma}$ by buying and selling the illiquid risky asset $S^2$, thereby ensuring $v<0$ on $\overline{\Gamma}$.
\end{proof}

\section{\bf Optimal policy}
To determine the optimal portfolio with transaction costs, we typically first identify the selling region, buying region, and no-trading region. In this context, the region $\Gamma$ is divided into three wedges $SR$, $NT$, and $BR$ by two rays passing through the origin based on the value function $v$. It can be shown that all three regions are non-empty. The wedge $SR$ represents the selling region, $NT$ represents the no-trading region, and $BR$ represents the buying region for the illiquid risky asset $S^2$.

The value function $v$ satisfies Proposition \ref{propo}, which includes homotheticity, continuity, and convexity. As a result, the conclusions of Chapter 6 in \cite{Sherve1994} are established. However, as the HJB equation (\ref{fxxhjb}) is not linear, improving the second-order smoothness using the viscosity solution theory is particularly challenging. Fortunately, \cite{hobson2019multi} established an important result when $\theta=1$ by studying a boundary value problem for a first-order differential equation. This result in \cite{hobson2019multi} is actually true for any $\theta$, and we can modify the form of $U$ in this literature to apply it. Following \cite{hobson2019multi}, we have the following theorem.
\begin{theorem}
	The value function $v$ is the unique solution of the HJB equation (\ref{hjb}).\label{C2}
\end{theorem} 
\begin{proof}
	The proof is similar to Theorem 4.1 and Theorem 4.3 in \cite{hobson2019multi} and the update of the utility function does not make the proof more difficult. 
	
\end{proof}

Theorem \ref{C2} implies that $v$ is twice continuously differentiable, i.e., $v\in \mathcal{C}^2$. Hence, the main results in Section6 Convex analysis of the value function in \cite{Sherve1994} can be summarized and presented in the following theorem without providing a proof.

\begin{theorem}\label{convexanalysis}
	%%	Define the subdifferential 
	%%	$$\partial v\left(x,y\right):=\{\left(\delta_x,\delta_y\right)\in\mathcal{R}^2|v\left(\xi,\eta\right)\le v\left(x,y\right)+\delta_x\left(\xi-x\right)+\delta_y\left(\eta-y\right),\forall \left(\xi,\eta\right)\in\Gamma\}$$
	The value function associated with Problem \eqref{problem} exhibits the following properties.
	\item[(1)]
	%%		-\left(1-\mu\right)\delta_x+\delta_y\ge0, \; \delta_x-\left(1-\lambda\right)\delta_y\ge0,\;\delta_x>0,\;\delta_y>0,\;\nonumber\\ \forall\left(\delta_x,\delta_y\right)\in\partial v\left(x,y\right),\;
	\begin{eqnarray}  v_x>0,\;v_y>0,\;
		\forall\left(x,y\right)\in\Gamma\nonumber.
	\end{eqnarray}
	\item[(2)]$\Gamma$ can be partitioned into three wedges, SR, BR, and NT (which may be an empty set), by two rays passing through $(0,0)$, where
	\begin{eqnarray}
		&&-\left(1-\mu\right)v_x+v_y=0, \;\forall\left(x,y\right)\in SR,\nonumber\\ &&v_x-\left(1-\lambda\right)v_y=0,\;\forall\left(x,y\right)\in BR,\nonumber\\
		&&-\left(1-\mu\right)v_x+v_y>0, \;v_x-\left(1-\lambda\right)v_y>0,\overline{\mathcal{L}}v=0,\;
		\forall\left(x,y\right)\in NT.	\nonumber
	\end{eqnarray}
	\item[(3)] There exist constants $a>0$ and $b>0$ such that
	\begin{eqnarray}
		v\left(x,y\right)&=&\frac{a}{p}\left(x+\left(1-\mu\right)y\right)^p,\quad\forall\left(x,y\right)\in SR,\nonumber\\
		v\left(x,y\right)&=&\frac{a}{p}\left(x+\frac{y}{1-\lambda}\right)^p,\quad\forall\left(x,y\right)\in BR.\nonumber
	\end{eqnarray}
\end{theorem}
These three regions represent the different cases in which the minimum is taken in the HJB equation. Additionally, they also have a more practical interpretation. The form of the value function $v$ in Theorem \ref{convexanalysis}-(3) implies that the agent will immediately sell the illiquid risky asset $S^2$ to reach $\overline{NT}$ in $SR$, and immediately sell the illiquid risky asset $S^2$ to reach $\overline{NT}$ in $BR$, which is why $SR$ is known as the selling region and $BR$ is known as the buying region. The natural question arises whether all three regions exist. Theorem \ref{v=phi} shows that $SR$ and $BR$ are non-empty when $0<p<1$. To discuss further results, we first transform the value function $v(x,y)$ into a univariate function.

To simplify the analysis, we can express the value function $v(x,y)$ as a function of a single variable $z:=\frac{y}{x+y}$, where $z$ represents the proportion of wealth invested in the illiquid risky asset $S^2$. We define the interval $I:=\left(-\frac{1-\lambda}{\lambda},1\right)$ and
$$u\left(z\right) = v\left(1-z,z\right),\forall z\in \overline{I}\quad i.e.,\quad v\left(x,y\right)=\left(x+y\right)^pu\left(\frac{y}{x+y}\right), \forall\left(x,y\right)\in\overline{\Gamma}\backslash\{\left(0,0\right)\}.$$

Because $u(z) = v(1-z,z)$, the function $u$ inherits concavity from $v$. Furthermore, $u\in\mathcal{C}^2$ and $u$ is the unique classical solution of the following second-order differential equation.

\begin{eqnarray}
	\!\!\!\!\min\!\!\!\!&&\left\{\beta u-p\left[r+\alpha_2z-\frac{\sigma_2^2}{2}\left(1-p\right)z^2\right]u-\left[\alpha_2z\left(1-z\right)-\sigma_2^2\left(1-p\right)z^2\left(1-z\right)\right]u'\right.\nonumber\\
	&&-\frac{1}{2}\sigma_2^2z^2\left(1-z\right)^2u''+\frac{1-\theta p}{p}\theta^{\frac{\theta p}{1-\theta p}}\left(1-z\right)^\frac{\left(1-\theta\right)p}{1-\theta p}\left(pu-zu'\right)^\frac{\theta p}{\theta p-1}\nonumber\\
	&&-\max_{\pi}\Big\{\frac{\pi^2}{2}\sigma_1^2\left[-p\left(1-p\right)\left(1-z\right)^2u+2\left(1-p\right)\left(1-z\right)^2zu'+z^2\left(1-z\right)^2u''\right]\nonumber\\
	&&+\rho\sigma_1\sigma_2\pi\left[-p\left(1-p\right)z\left(1-z\right)u+\left(1-p\right)z\left(1-z\right)\left(2z-1\right)u'-z^2\left(1-z\right)^2u''\right]\nonumber\\
	&&+\alpha_1\pi\left[p\left(1-z\right)u-z\left(1-z\right)u'\right]\Big\}, pu+\frac{1}{\mu}\left(1-\mu z\right)u',\left.pu-\frac{1}{\lambda}\left(1-\lambda\left(1-z\right)\right)u'\right\}\nonumber\\&=&0.\label{newhjb}
\end{eqnarray}
\begin{remark}
	In the HJB equation (\ref{newhjb}), we can express the term $$\frac{1-\theta p}{p}\theta^{\frac{\theta p}{1-\theta p}}\left(1-z\right)^\frac{\left(1-\theta\right)p}{1-\theta p}\left(pu-zu'\right)^\frac{\theta p}{\theta p-1}$$   as
	$$ \min_{d}\left\{\left(pu-zu^{\prime}\right)-\frac{\left(d^{\theta}z^{1-\theta}\right)^p}{p}\right\},$$
	where the minimum point is attained at $d^*=\frac{c}{x+y}$. This expression is more suitable for numerical computation, and we will use it in Section \ref{numans}.
\end{remark}
This paper introduces a liquid risk asset and incorporates liquidity preference. Although we limit the solvency region $\Lambda$ to $\Gamma$ with liquidity preference, all the three wedges $NT$, $SR$, and $BR$ are non-empty.
\begin{proposition}\label{sanquyufeikong}
	Whether $p<0$ or $0<p<1$, all the three wedges are non-empty.
	\item[(1)] $NT\neq \emptyset$. 
	%% Assume that$$\left[\beta-rp-\frac{p\alpha_1^2}{2\left(1-p\right)\sigma_1^2}\right]\frac{1-\theta}{1-\theta p} +\rho\frac{\sigma_2}{\sigma_1}\alpha_1-\alpha_2 \neq 0,$$ 
	\item[(2)] $SR\neq \emptyset$.
	\item[(3)] $BR\neq \emptyset$.	
\end{proposition}
\begin{proof}
	
	\item[(1)]
	When $p<0$, we first claim that $SR\neq \Gamma$ and $BR\neq\Gamma$. 
	
	If $SR = \Gamma$, according to continuity of $v$ and Theorem \ref{convexanalysis}-(3), $v$ must be finite on the $\partial_2\Gamma$, which is a contradiction. On the other hand, if $BR=\Gamma$, we can always buy illiquid risky asset $S^2$ in such a way that liquid wealth $X$ becomes zero. In this case, the value function $v$ would not change, and it would be infinite on $\partial_1\Gamma$, which is also a contradiction.
	
	Proposition \ref{feikong} demonstrates that $SR$ and $BR$ are both non-empty when $0<p<1$.
	
	Whenever $p<0$ or $0<p<1$, if $NT = \emptyset$, then $SR$ and $BR$ must share a common boundary $H$ in $\Gamma$. However, according to the definition of $SR$ and $BR$, we know $v_x=v_y=0$ on $H$, which contradicts Theorem \ref{convexanalysis}-(1).
	
	\item[(2)] We only need to consider the case when $p<0$. It is worth noting that $u$ is finite at $z=1$, and we have
	$$ (1-z)u^{\prime}\left(z\right)\rightarrow 0,\quad (1-z)^2u^{\prime\prime}\left(z\right)\rightarrow 0, \quad as\quad z\rightarrow 1. $$
	If $SR= \emptyset$, then as $z\rightarrow 1 $, Eq.~(\ref{newhjb}) implies $$\lim\limits_{z\rightarrow 1}v_x\left(1-z,z\right)=\lim\limits_{z\rightarrow 1}pu\left(z\right)-zu^{\prime}\left(z\right)=0.$$
	However, this contradicts
	\begin{eqnarray}
		v_x\left(1-z,z\right)&=&\lim\limits_{h\downarrow 0}\frac{1}{h}\left[v\left(1-z+h,z\right)-v\left(1-z,z\right)\right]\nonumber
		\\ &=&\lim\limits_{h\downarrow 0}\frac{1}{h}\left[\left(\frac{h}{1-z+\frac{z}{1-\lambda}}+1\right)^p-1\right]v\left(0,1\right)\nonumber\\
		&=& \left(1-z+\frac{1}{1-\lambda}\right)^{-1}pv\left(1-z,z\right),\nonumber
	\end{eqnarray}
	so that $$\lim\limits_{z\rightarrow 1}v_x\left(1-z,z\right)=\left(1-\lambda\right)pv\left(0,1\right)>0.$$
	
	\item[(3)]It is necessary to consider only the case where $p<0$. If $BR= \emptyset$, then we can show that $u$ satisfies the first term of HJB equation (\ref{newhjb}) near $z=-\frac{1-\lambda}{\lambda}$. Specifically, we can express the equation as a quadratic function of $u^{\prime\prime}$ of the form
	$$-d_1\left(z\right)\left(u^{\prime\prime}\right)^2+d_2\left(z,u,u^{\prime}\right)u^{\prime\prime}+d_3\left(z,u,u^{\prime}\right)=0 ,$$
	where $$d_1\left(z\right)=\left(1-\rho^2\right)\sigma_1^2\sigma_2^2z^4\left(1-z\right)^4.$$
	Moreover,
	$$\frac{1-\theta p}{p}\theta^{\frac{\theta p}{1-\theta p}}\left(1-z\right)^\frac{\left(1-\theta\right)p}{1-\theta p}t^\frac{\theta p}{\theta p-1}$$ is Lipschitz continuous about $t$ on any half-line of the form $\left[\gamma, \infty\right)$, where $\gamma >0$. Note that $$pu-zu^{\prime}=v_x\left(1-z,z\right)\ge pv\left(1-z,z\right)\left(1-z+\frac{z}{1-\lambda}\right)^{-1},$$ so that $pu-zu^{\prime}$ tends to $+\infty$ as $z\rightarrow -\frac{1-\lambda}{\lambda}$, implying that $d_2$ and $d_3$ are Lipschitz continuous with respect to $u$ and $u^{\prime}$ when $z$ is near $-\frac{1-\lambda}{\lambda}$. By the quadratic formula, we can write $$u^{\prime\prime}=F\left(z,u,u^{\prime}\right),$$ where $F$ is Lipschitz continuous with respect to $u$ and $u^{\prime}$ when $z$ is near $-\frac{1-\lambda}{\lambda}$. Hence, $\lim\limits_{z\rightarrow -\frac{1-\lambda}{\lambda}}u\left(z\right)$ exists and is finite, which leads to a contradiction.
\end{proof}

According to Theorem \ref{convexanalysis}-(2) and Proposition \ref{sanquyufeikong}, it can be shown that there exist two numbers $\eta_2$ and $\eta_1$, with $-\frac{1-\lambda}{\lambda}<\eta_2<\eta_1<1$, such that
\begin{eqnarray}
\partial_1NT&:=&\overline{SR}\cap\overline{NT}=\{\left(x,y\right)\mid x\ge 0, y=\frac{\eta_1}{1-\eta_1}x\},\nonumber\\
\partial_2NT&:=&\overline{BR}\cap\overline{NT}=\{\left(x,y\right)\mid x\ge 0, y=\frac{\eta_2}{1-\eta_2}x\}.\nonumber
\end{eqnarray}
The two rays $\partial_1NT$ and $\partial_2NT$ partition the solvency region $\Gamma$ into three wedges: $SR$, $NT$, and $BR$, see Fig.~\ref{figure2} below.
\begin{figure}[htbp]
	\centering
	\includegraphics[scale=0.45]{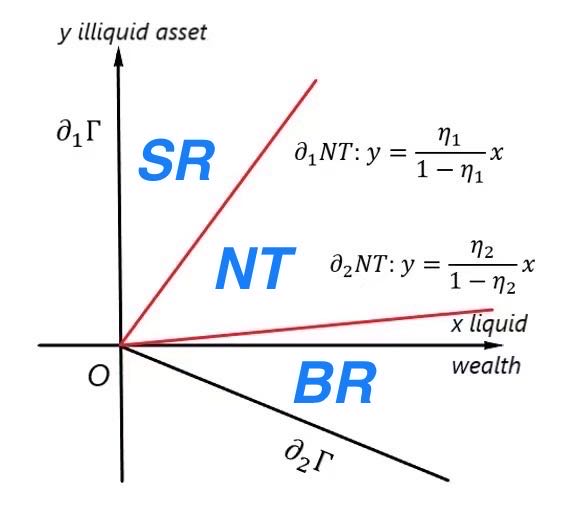}
	\caption{Three regions: $SR$, $NT$ and $BR$.}
	\label{figure2}
\end{figure}

Define the reflection direction on $\partial NT$
$$\gamma\left(x,y\right)=\left(\gamma_1\left(x,y\right),\gamma_2\left(x,y\right)\right):=\left\{\begin{aligned}
	\left(1-\mu,-1\right), \qquad if \left(x,y\right)\in\partial_1 NT,\\
	\left(-1,1-\lambda\right), \qquad if \left(x,y\right)\in\partial_2 NT.
\end{aligned}\right.$$

If the optimal policy exists, i.e., the problem is reachable, then the solvency region $\Gamma$ can be divided into three regions based on the value function $v$, and the optimal policy can be similarly divided. In particular, the HJB equation in the no-trading region $NT$ implies that the optimal policy does not involve trading the illiquid risky asset $S^2$. This region is therefore referred to as the no-trading region. The optimal values of $\left(c^*,\pi^*\right)$ in the no-trading region satisfy the following
\begin{eqnarray}\label{c*1}
	c_t^*&=&\theta^{\frac{1}{1-\theta p}}X_t^{\frac{\left(1-\theta\right)p}{1-\theta p}}v_x\left(X_t,Y_t\right)^{-\frac{1}{1-\theta p}}, \\
\label{pi*}
	\pi_t^*&=&-\frac{\alpha_1v_x\left(X_t,Y_t\right)+\rho\sigma_1\sigma_2Y_tv_{xy}\left(X_t,Y_t\right)}{\sigma_1^2X_tv_{xx}\left(X_t.Y_t\right)}.
\end{eqnarray}
Furthermore, initial position $\left(X_{0-},Y_{0-}\right)$ in $SR$ and $BR$ must immediately jump to the region $\overline{NT}$, and any point $\left(X_t,Y_t\right)$ reaching the boundaries of the no-trading region $\partial_1NT$ and $\partial_2NT$ must be reflected back into the region $NT$.

We will proceed in two steps: first, we will prove the existence of the policy described above; second, we will show its optimality. Without loss of generality, we can assume $(x_0,y_0)\in\overline{NT}$, as we can choose appropriate values of $\left(L^*_0, M^*_0\right)$ such that $(X_0,Y_0)\in\partial NT$ if $(x_0,y_0)\notin\overline{NT}$.
\begin{lemma}\label{op1}
	Assume $(x_0,y_0)\in\overline{NT}$, then there exist continuous processes $X^*$, $Y^*$ and $k$ such that $X^*_0=x_0$, $Y^*_0=y_0$, $k_0=0$ and 
	\begin{eqnarray}
		(X^*_t,Y^*_t)\in\overline{NT}, \quad \forall t\ge 0,\nonumber
	\end{eqnarray}
	\begin{eqnarray}
		dX^*_t&=&\left[\left(r+\alpha_1\pi_t^*\right)X^*_t-c_t^*\right]dt + \sigma_1\pi_t^*X^*_tdB^1_t+\gamma_1\left(X^*_t,Y^*_t\right)dk_t,\nonumber\\
		dY^*_t&=&\left(r+\alpha_2\right)Y^*_tdt+\sigma_2Y^*_tdB^2_t+\gamma_2\left(X^*_t,Y^*_t\right)dk_t,\nonumber\\
		k_t&=&\int_{0}^{t}1_{\{\left(X^*_t,Y^*_t\right)\in\partial NT\}}dk_t,\nonumber
	\end{eqnarray}
	where $\left(c^*,\pi^*\right)$ is shown as Eqs.~(\ref{c*1}) and (\ref{pi*}).
	
	Then the processes $L$ and $M$ have the following
	
	\begin{eqnarray}
		L^*_t=L^*_0+\int_{0}^{t}1_{\{\left(X^*_t,Y^*_t\right)\in\partial_2 NT\}}dk_t,\nonumber
	\end{eqnarray}
	\begin{eqnarray}
		M^*_t=M^*_0+\int_{0}^{t}1_{\{\left(X^*_t,Y^*_t\right)\in\partial_1 NT\}}dk_t,\nonumber
	\end{eqnarray}
	where $L^*_0=0, M^*_0=0$.
\end{lemma}

\begin{proof}
	
	The proof is based on a similar approach to that used in Lemma 9.3 of \cite{Sherve1994}, but we need to make some modifications to account for our specific models. These modifications are as follows:
	
	First, we assume that $v$ is $\mathcal{C}^2$ on the entire region $\Gamma$, so that the proof holds for any $-\frac{1-\lambda}{\lambda}<\eta_2<\eta_1<1$.
	
	Second, we define the constant $m_2$ as $$m_2:=\max_{\eta_2<z<\eta_1} u_q\left(z,1-z\right),$$ instead of using the value given in \cite{Sherve1994}.
	
	Moreover, in the proof of $\rho\le\tau$, we replace $t\wedge \rho_n\wedge\tau_n$ with $t\wedge \rho_n\wedge\tau_n\wedge \kappa_m$, where $$\kappa_m:=\inf\{t\ge0\mid|\pi_t|>m\}.$$ After taking the expectation of Eq.~(\ref{ito}), we let $m$ tend to $\infty$.
	
	Similarly, in the later stages of the proof, we replace $\eta_n$ with $\eta_n\wedge \kappa_m$, and also let $m$ tend to $\infty$ after taking the expectation of Eq.~(\ref{ito}).
\end{proof}

\begin{theorem}
	$\left(c^*,\pi^*,L^*,M^*\right)$ in Lemma \ref{op1}  is the optimal strategy solving Problem \eqref{problem}.
\end{theorem}
\begin{proof}
	Similar to the proof of Lemma 9.5 in \cite{Sherve1994}, there exist constants $m_3$ and $m_4$ such that
	\begin{eqnarray}
		v_x\left(x,y\right)\le m_3\left(x+y\right)^{p-1}, \quad\forall \left(x,y\right)\in\overline{NT},\nonumber
	\end{eqnarray}
	\begin{eqnarray}
		|yv_y\left(x,y\right)|\le m_4\left(x+y\right)^p, \quad\forall \left(x,y\right)\in\overline{NT}.\nonumber
	\end{eqnarray}
	Similarly, there exists a constant $m_5>0$ such that
	\begin{eqnarray}
		|xv_x\left(x,y\right)|\le m_5\left(x+y\right)^p, \quad\forall \left(x,y\right)\in\overline{NT}.\nonumber
	\end{eqnarray}
	Note that the equations $-\left(1-\mu\right)v_x\left(x,y\right)+v_y\left(x,y\right)=0$ and $v_x\left(x,y\right)-\left(1-\lambda\right)v_y\left(x,y\right)=0$ hold when $y=\frac{\eta_1}{1-\eta_1}x$ and $y=\frac{\eta_2}{1-\eta_2}x$, respectively. For any almost surely finite stopping time $\tau$, we have
	\begin{eqnarray}
		v\left(x_0,y_0\right)&=&e^{-\beta\tau}v\left(X^*_{\tau},Y^*_{\tau}\right)+\int_{0}^{\tau}e^{-\beta t}U\left(c^*_t,X^*_t\right) dt\nonumber\\
		&-& \int_{0}^{\tau}e^{-\beta t}\sigma_1\pi^*_tX^*_tv_x\left(X^*_{t},Y^*_{t}\right) dB^1_t -\int_{0}^{\tau}e^{-\beta t}\sigma_2Y^*_tv_y\left(X^*_{t},Y^*_{t}\right) dB^2_t.\label{optim-v}
	\end{eqnarray}
	Let $\kappa_m:=\inf\{t\ge0\mid|\pi^*_t|>m\}$, $\tau_n:=\inf\{t\ge0\mid X^*_t+Y^*_t\le\frac{1}{n}\}$, and $\tau_0:=\inf\{t\ge0\mid X^*_t=Y^*_t=0\}$. When $p<0$, we can replace $\tau$ with $t\wedge \tau_n \wedge \kappa_m$ in Eq.~(\ref{optim-v}), take the expectation, and then let $m\rightarrow \infty$ to obtain
	$$ v\left(x_0,y_0\right)\le \mathbb{E}\left[\int_{0}^{t\wedge\tau_n}e^{-\beta t}U\left(c^*_t,X^*_t\right) dt\right].$$
	In addition, we need to prove that $\tau_0=\infty$ almost surely. Note that
	\begin{eqnarray}
		\lim\limits_{n\rightarrow\infty}\lim\limits_{t\uparrow \tau_n}\left[e^{-\beta \left(t\wedge \tau_n\right)}v_x\left(X^*_{t\wedge \tau_n},Y^*_{t\wedge \tau_n}\right)+ \int_{0}^{t\wedge \tau_n}e^{-\beta s}U\left(c^*_s,X^*_s\right)ds\right]=-\infty\quad \mbox{on}\quad \left\{\tau_0<\infty\right\},\nonumber
	\end{eqnarray}
	we have
	\begin{eqnarray}
		\lim\limits_{n\rightarrow\infty}\lim\limits_{t\uparrow \tau_n}\left[\int_{0}^{t\wedge\tau_n} e^{-\beta s}\sigma_1\pi^*_sX^*_sv_x\left(X^*_{s},Y^*_{s}\right) dB^1_s +\int_{0}^{\tau} e^{-\beta s}\sigma_2Y^*_sv_y\left(X^*_{s},Y^*_{s}\right) dB^2_s\right]=-\infty\nonumber\\ \mbox{on}\quad \left\{\tau_0<\infty\right\} ,\label{=-infty} 
	\end{eqnarray}
	which implies
	\begin{eqnarray}
		\int_{0}^{\tau_0}e^{-2\beta t}\left\{\left[\sigma_1\pi^*_tX^*_tv_x\left(X^*_{t},Y^*_{t}\right)\right]^2 +\left[\sigma_2Y^*_tv_y\left(X^*_{t},Y^*_{t}\right)\right]^2\right.\nonumber\\
		\left.+2\rho\sigma_1\sigma_2\pi^*_tX^*_tY^*_tv_x\left(X^*_{t},Y^*_{t}\right)v_y\left(X^*_{t},Y^*_{t}\right) \right\} dt  = \infty.  \label{=infty}
	\end{eqnarray}
	However, because Eq.~(\ref{=infty}) implies that the limit in Eq.~(\ref{=-infty}) does not exist, we conclude $\tau_0=\infty$ almost surely. Finally, taking the limits $n\rightarrow \infty$ and $t\rightarrow \infty$, we obtain the optimality of $\left(C^*,\pi^*,L^*,M^*\right)$ when $p<0$.
	
	We now turn to the case when $0<p<1$. In this case, we define $Z_t = X^*_t + Y^*_t$. Note that $$\left(1-\eta_1\right)\left(x+y\right)\le x\le\left(1-\eta_2\right)\left(x+y\right)<\frac{1}{\lambda}\left(x+y\right)\quad \forall\left(x,y\right)\in\overline{NT}$$ 
	and
	$$|y|\le \max\{|\eta_1|,|\eta _2|\}\left(x+y\right):= m_6\left(x+y\right),\quad \forall\left(x,y\right)\in\overline{NT},$$ 
	we have
	\begin{eqnarray}
		Z_{t\wedge \kappa_m} &=& Z_0+\int_{0}^{t\wedge \kappa_m}\left[\left(r+\alpha_1\pi^*_s\right)X^*_s+\left(r+\alpha_2\right)Y^*_s-c^*_s\right]ds\nonumber\\
		&+&\sigma_1\int_{0}^{t\wedge \kappa_m}\pi^*_sX^*_sdB^1_s
		+\sigma_2\int_{0}^{t\wedge \kappa_m}Y^*_sdB^2_s
		-\lambda L_s -\mu M_s\nonumber
		\\
		&\le& Z_0+m_7\int_{0}^{t\wedge \kappa_m}Z_sds +\sigma_1 H_{t\wedge \kappa_m} + N_{t\wedge \kappa_m},\nonumber
	\end{eqnarray}
	where $$m_7=r+\frac{\alpha_1 m}{\lambda}+m_6\alpha_2,$$ $$H_t=\int_{0}^{t}\pi^*_sX^*_sdB^1_s$$ and $$N_t=\int_{0}^{t}Y^*_sdB^2_s.$$ 
	Doob's maximal martingale inequality yields 
	\begin{eqnarray}
		\mathbb{E}\left(H^*_{t\wedge \rho_n\wedge \kappa_m}\right)^2\le 4\mathbb{E}\left[H^2_{t\wedge \rho_n\wedge \kappa_m}\right]=4\mathbb{E}\left[\int_{0}^{t\wedge\rho_n\wedge \kappa_m}\pi^{*2}_sX^{*2}_sds\right]\le 4\frac{m^2}{\lambda^2}\mathbb{E}\left[\int_{0}^{t\wedge\rho_n\wedge \kappa_m}Z^2_sds\right],\nonumber
	\end{eqnarray}
	where $$\rho_n=\inf\{t\ge 0\mid Z_t\ge n\},$$ $$Z^*_t=\max_{0\le s\le t}Z_s,\quad H^*_t=\max_{0\le s\le t}|H_s|.$$ Similarly, define $$N^*_t=\max_{0\le s\le t}|N_s|,$$ and then
	\begin{eqnarray}
		\mathbb{E}\left(N^*_{t\wedge \rho_n\wedge \kappa_m}\right)^2\le  4m^2_6\mathbb{E}\left[\int_{0}^{t\wedge\rho_n\wedge \kappa_m}Z^2_sds\right].\nonumber
	\end{eqnarray}
	Using H\"{o}lder's inequality, we find some $m_8>0$ such that for every $T>0$, 
	\begin{eqnarray}
		\mathbb{E}\left(z^*_{t\wedge \rho_n\wedge \kappa_m}\right)^2&\le& m_8\left[Z^2_0+\mathbb{E}\left(\int_{0}^{t\wedge\rho_n\wedge \kappa_m}Z^*_sds\right)^2+\mathbb{E}\left(H^*_{t\wedge \rho_n\wedge \kappa_m}\right)^2+\mathbb{E}\left(N^*_{t\wedge \rho_n\wedge \kappa_m}\right)^2\right]\nonumber\\
		&\le& m_8\left[Z^2_0+\left(T+4\frac{m^2}{\lambda^2}+4m^2_6\right)\int_{0}^{t\wedge\rho_n\wedge \kappa_m}\mathbb{E}\left(Z^*_s\right)^2ds\right],\quad\forall t\in\left[0,T\right].\nonumber
	\end{eqnarray}
	According to Gronwall's inequality, we have
	\begin{eqnarray}
		\mathbb{E}\left(z^*_{t\wedge \rho_n\wedge \kappa_m}\right)^2\le m_8Z_0^2\exp\left[m_8\left(T+4\frac{m^2}{\lambda^2}+4m^2_6\right)t\right],\quad\forall t\in\left[0,T\right].\nonumber
	\end{eqnarray}
	Taking the limit $n\rightarrow \infty$ and setting $t=T$, we obtain
	\begin{eqnarray}
		\mathbb{E}\left(X^*_{T\wedge \kappa_m}+Y^*_{T\wedge \kappa_m}\right)^2\le m_8Z_0^2\exp\left[m_8\left(T+4\frac{m^2}{\lambda^2}+4m^2_6\right)T\right],\quad\forall T\ge 0\nonumber.
	\end{eqnarray}
	Therefore, we have 
	\begin{eqnarray}
		\mathbb{E}\left[\int_{0}^{T\wedge \kappa_m}\left[X^*_s+Y^*_s\right]^2ds\right]<\infty,\quad\forall T\ge0\nonumber,
	\end{eqnarray}
	which implies
	\begin{eqnarray}
		\mathbb{E}\left[\int_{0}^{T\wedge \kappa_m}\left[\sigma_1\pi^*_sX^*_sv_x\left(X^*_{s},Y^*_{s}\right)\right]^2ds\right]<\infty,\quad \mathbb{E}\left[\int_{0}^{T\wedge \kappa_m}\left[\sigma_2Y^*_sv_y\left(X^*_{s},Y^*_{s}\right)\right]^2ds\right]<\infty,\quad\forall T\ge0\nonumber.
	\end{eqnarray}
	
	From Eq.~(\ref{optim-v}), we have
	\begin{eqnarray}
		v\left(x_0,y_0\right)&=&e^{-\beta t\wedge \kappa_m}v\left(X^*_{t\wedge \kappa_m},Y^*_{t\wedge \kappa_m}\right)+\int_{0}^{t\wedge \kappa_m}e^{-\beta s}U\left(c^*_s,X^*_s\right) ds\nonumber,
	\end{eqnarray}
	Taking the limit $m\rightarrow \infty$, we obtain
	\begin{eqnarray}
		v\left(x_0,y_0\right)&=&\mathbb{E}\left[e^{-\beta t}v\left(X^*_{t},Y^*_{t}\right)\right]+\mathbb{E}\left[\int_{0}^{t}e^{-\beta s}U\left(c^*_s,X^*_s\right) ds\right].\label{gujijieguo}
	\end{eqnarray}	
	From Eq.~(\ref{c*1}), we have
	\begin{eqnarray}
		c^{*\theta}_tX^{*1-\theta}_t&=&\theta^{\frac{\theta}{1-\theta p}}X_t^{*\frac{1-\theta}{1-\theta p}}v_x\left(X^*_t,Y^*_t\right)^{-\frac{\theta}{1-\theta p}}\nonumber\\
		&\ge&\theta^{\frac{\theta}{1-\theta p}}\left(1-\eta_1\right)^{\frac{1-\theta}{1-\theta p}}m_3^{-\frac{\theta}{1-\theta p}}\left(X^*_t+Y^*_t\right)\nonumber\\
		&:=& m_9\left(X^*_t+Y^*_t\right).\nonumber
	\end{eqnarray}
	Then
	\begin{eqnarray}
		\int_{0}^{\infty}\mathbb{E}\left[e^{-\beta t}\left(X^*_{t}+Y^*_{t}\right)^p\right] dt\le\frac{p}{m^p_9}\int_{0}^{\infty}\mathbb{E}\left[e^{-\beta t}U\left(c^*_t,X^*_t\right)\right] dt < \infty\nonumber.
	\end{eqnarray}
	Thus, there exists a sequence $t_n\uparrow\infty$ such that
	\begin{eqnarray}
		\lim\limits_{n\rightarrow \infty}\mathbb{E}\left[e^{-\beta t_n}\left(X^*_{t_n}+Y^*_{t_n}\right)^p\right] = 0\nonumber.
	\end{eqnarray}
	It follows that
	\begin{eqnarray}
		\lim\limits_{n\rightarrow \infty}\mathbb{E}\left[e^{-\beta t_n}v\left(X^*_{t_n},Y^*_{t_n}\right)\right]\le \max_{\eta_2<z<\eta_1}u(z)\times\lim\limits_{n\rightarrow \infty}\mathbb{E}\left[e^{-\beta t_n}\left(X^*_{t_n}+Y^*_{t_n}\right)^p\right] = 0.
	\end{eqnarray}
	Substituting $t$ with $t_n$ in Eq.~(\ref{gujijieguo}) and taking the limit as $n\rightarrow\infty$, we obtain the optimality of $\left(C^*,\pi^*,L^*,M^*\right)$.
\end{proof}

\section{\bf Numerical analysis}
\label{numans}
In this section, we use numerical methods to solve the HJB equation \eqref{zyshjb} and analyze the results. We follow the numerical calculation method described in Section5 Numerical methods in \cite{akian1996investment}. Our objective is two-fold: first, we investigate the effect of various parameters on the location of the no-trading region $NT$ in the optimal policy; second, we examine the evolution of the investment and consumption ratios on $NT$. We assume that the time discount rate $\beta$ is higher than the risk-free rate $r$ and that the two risk assets are positively correlated, and we conduct numerical tests with parameter values $\beta=0.1$, $r=7\%$, $\rho=0.4$, $\lambda=0.2$, and $\mu=0.2$.

\subsection{\bf Optimal trading boundaries}
We begin by investigating the relationship between the buying-selling boundaries $\left(\eta_2,\eta_1\right)$ associated with $\partial NT$ and other parameters, such as $\left(\alpha_1,\sigma_1\right)$, $\left(\alpha_2,\sigma_2\right)$, the risk aversion parameter $p$, and the liquidity preference parameter $\theta$.

We compare four scenarios to investigate the effects of $S^1$ with parameters $\left(\alpha_1,\sigma_1\right)$ and $S^2$ with parameters $\left(\alpha_2,\sigma_2\right)$:

\begin{enumerate}[itemindent=1em]
	\item[scenario 1:] $\left(\alpha_1,\sigma_1\right)=\left(4\%,30\%\right)$, $\left(\alpha_2,\sigma_2\right)=\left(8\%,35\%\right)$.
	\item[scenario 2:] There is no $\left(\alpha_1,\sigma_1\right)$, and $\left(\alpha_2,\sigma_2\right)=\left(8\%,35\%\right)$.
	\item[scenario 3:] $\left(\alpha_1,\sigma_1\right)=\left(4\%,30\%\right)$, $\left(\alpha_2,\sigma_2\right)=\left(13\%,35\%\right)$.
	\item[scenario 4:] $\left(\alpha_1,\sigma_1\right)=\left(4\%,30\%\right)$, $\left(\alpha_2,\sigma_2\right)=\left(8\%,50\%\right)$.
\end{enumerate}

We analyze both the cases $p>0$ and $p<0$. In scenario 1, we plot the buy-sell boundaries (illiquid asset-total wealth ratio) $\left(\eta_2,\eta_1\right)$ in Fig.~\ref{gra1}. We observe that both the buying boundary $\eta_2$ and the selling boundary $\eta_1$ increase with an increase in $p$ or $\theta$. This indicates that the width of the buying region $BR$ increases with $\theta$ and $\rho$, while the width of the selling region $SR$ decreases with $\theta$ and $\rho$. Agents with higher risk aversion tend to hold more liquid wealth, and the introduction of liquidity preference indeed encourages the agent to hold more liquid wealth. Therefore, a more liquidity-preferred (smaller $\theta$) or risk-averse (smaller $p$)  agent holds more of the  liquid wealth. We also observe differences in the speed of changes for positive and negative $\eta_1$ and $\eta_2$. When $p$ or $\theta$ decreases and $\eta_1$ is already close to 0, $\eta_1$ decreases slowly, while $\eta_2$ decreases much faster. As long as $\theta\neq 1$, we have $\eta_1<1$, which implies the existence of $SR$. This result is consistent with Conclusion (2) in Proposition \ref{sanquyufeikong}. However, when $\theta=1$, it is possible that $\eta_1=1$ and $SR$ does not exist. The product form of the utility function with $\theta<1$ guarantees the existence of $SR$. In contrast to the results in \cite{Sherve1994}, $\eta_2$ may be less than 0, and even $\eta_1$ may be less than 0, which means that the intersection of $NT$ (and even $SR$) and the fourth quadrant may not be empty. This phenomenon does not occur when $\theta=1$, and  this difference is mainly due to the liquidity preference in the utility function.
\begin{figure}[htbp]
	\centering
	\includegraphics[scale=0.57]{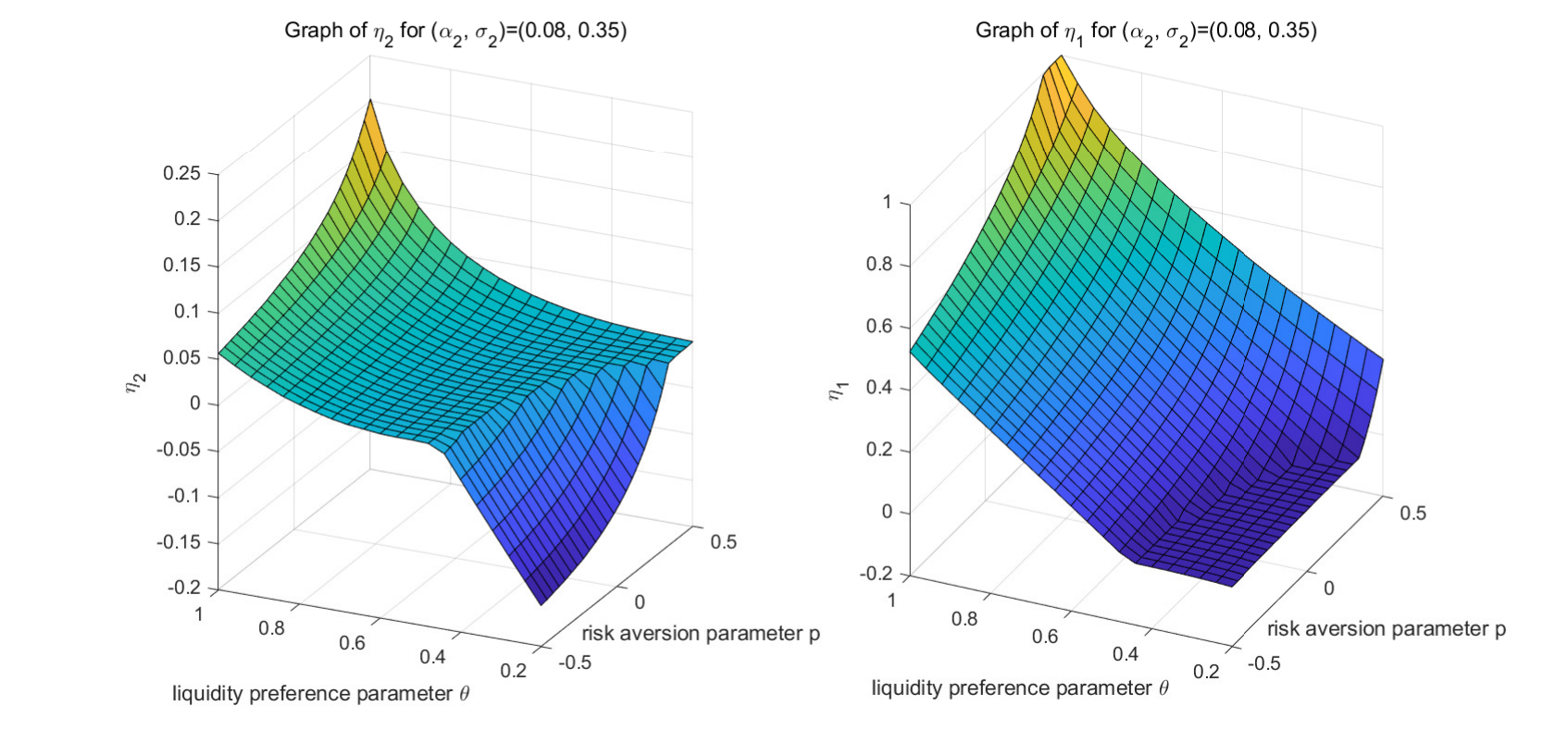}
	\caption{Buy-sell boundaries in scenario 1.}
	\label{gra1}
\end{figure}

\begin{figure}[htbp]
	\centering
	\includegraphics[scale=0.57]{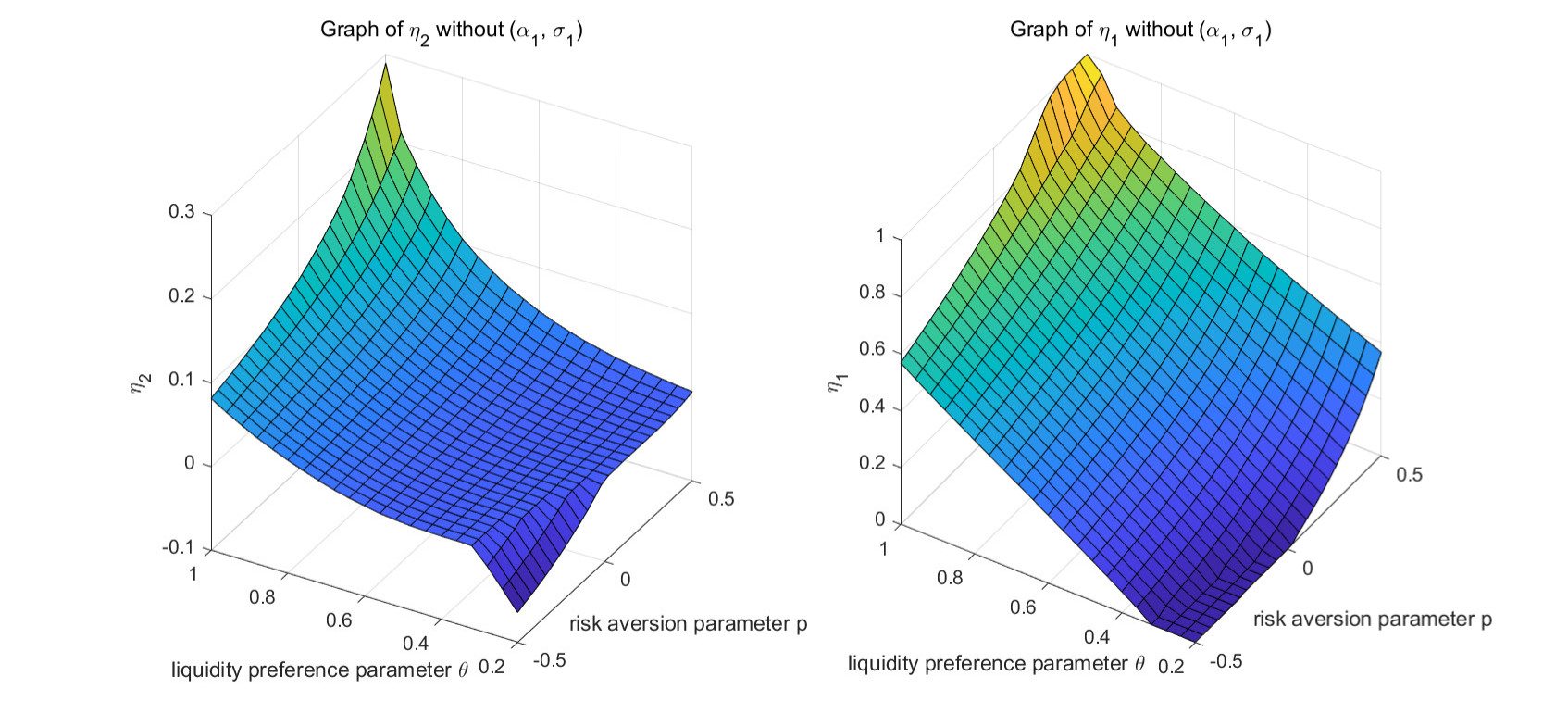}
	\caption{Buy-sell boundaries in scenario 2.}
	\label{gra2}
\end{figure}

\begin{figure}[htbp]
	\centering
	\includegraphics[scale=0.57]{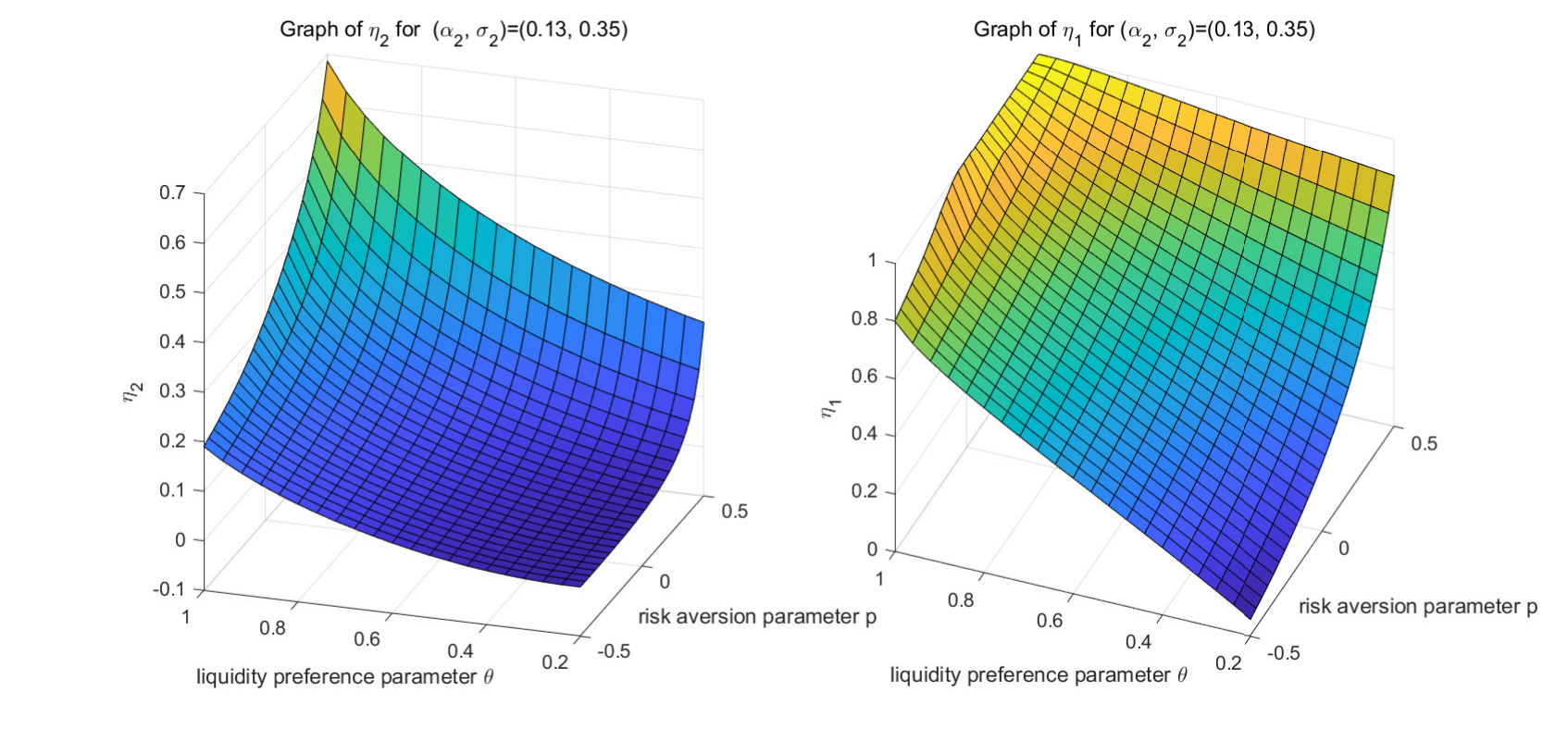}
	\caption{Buy-sell boundaries in scenario 3.}
	\label{gra3}
\end{figure}

\begin{figure}[htbp]
	\centering
	\includegraphics[scale=0.57]{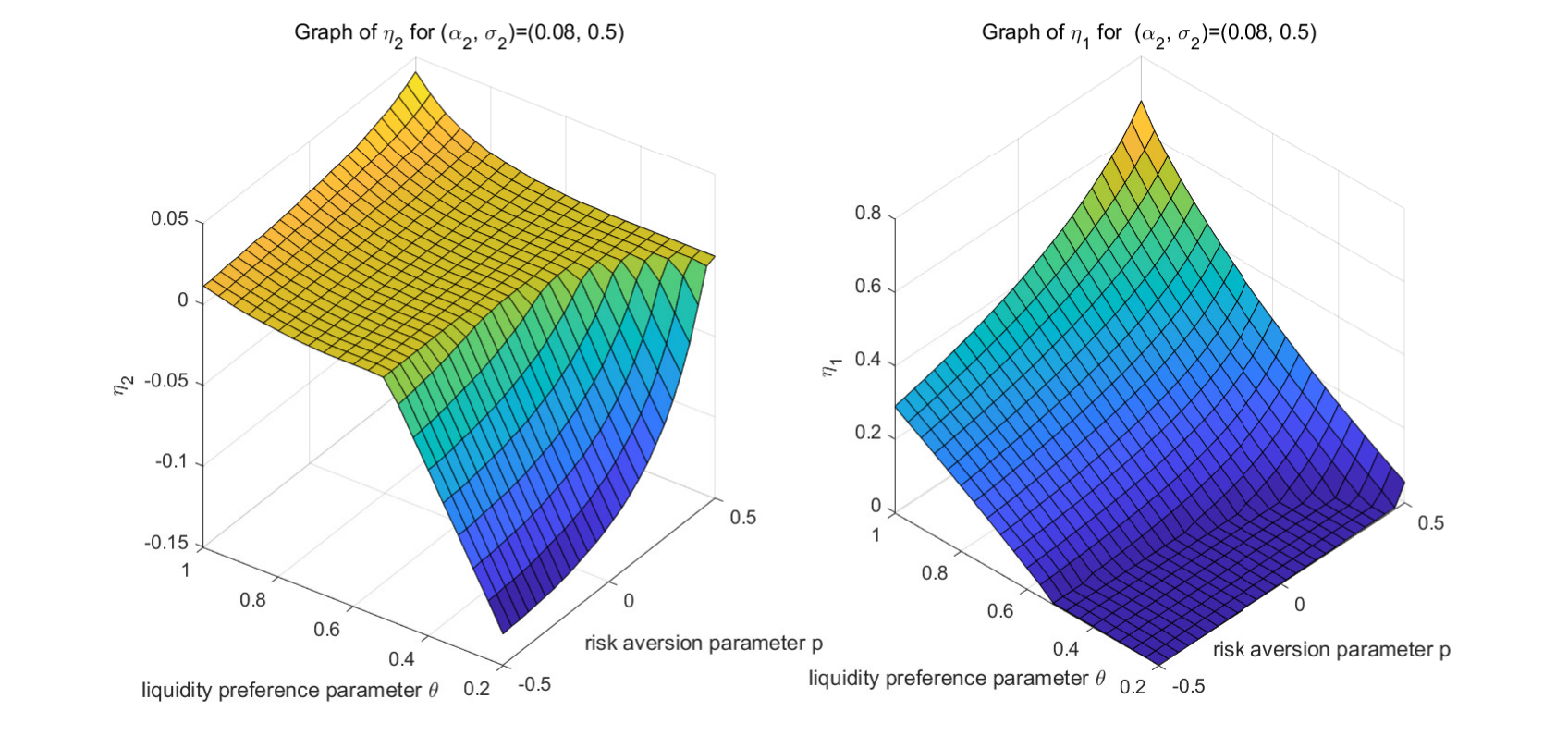}
	\caption{Buy-sell boundaries in scenario 4.}
	\label{gra4}
\end{figure}
We present the buy-sell boundaries $\left(\eta_2,\eta_1\right)$ for scenarios 2-4 in Fig.~\ref{gra2}-\ref{gra4}. The changes in the buy-sell boundary with respect to $\theta$ and $p$ are similar across the different scenarios.

Comparing Fig.~\ref{gra1} and Fig.~\ref{gra2}, we observe that the introduction of the liquid risk asset $S^1$ has a certain influence on the buy-sell boundary. When the liquid assets include the liquid risk asset $S^1$ in addition to the risk-free bond $S^0$, the potential return on liquid assets has risen. This leads to a higher preference for holding liquid wealth, which results in larger values of $\eta_2$ and $\eta_1$ in Fig.~\ref{gra2} compared to those in Fig.~\ref{gra1}.

Comparing Fig.~\ref{gra1} and Fig.~\ref{gra3}, we observe that if the illiquid risky asset $S^2$ has a higher expected rate of return $\alpha_2$, the agent may prefer to allocate more of their wealth to $S^2$, as indicated by the significant improvement in $\eta_2$ and $\eta_1$ values. Naturally, as the expected return on the illiquid asset becomes more attractive, an agent would like to hold more illiquid asset.  

Comparing Fig.~\ref{gra1} and Fig.~\ref{gra4}, we observe that the effect of the volatility coefficient $\sigma_2$ of the illiquid risky asset $S^2$ on the buy-sell boundary is complicated. For agents with high risk aversion and high liquidity preference, if $\eta_1$ is already close to 0 and $\eta_2$ is less than 0, a larger $\sigma_2$ may lead to slightly higher values of $\eta_2$ and $\eta_1$, indicating a preference for holding more of the illiquid asset for larger $\sigma_2$. However, for other agents, a larger $\sigma_2$ may lead to a decrease in $\eta_2$ and $\eta_1$, making them more reluctant to invest in the illiquid risky asset. In practice, an appropriate value of $\theta$ should be close to 1, and thus we should pay more attention to the conclusion where $\theta$ approaches 1, i.e., holding less of the illiquid asset as the volatility increases.

Moreover, the illiquid risky asset $S^2$ has a greater impact on the buy-sell boundaries $\eta_2$ and $\eta_1$ compared to the liquid risk asset $S^1$. In addition, $\eta_2$ is more sensitive to changes in the parameters related to $S^1$ and $S^2$ than $\eta_1$, indicating that the buy boundary is more affected by changes in the expected returns and volatilities of $S^1$ and $S^2$. On the other hand, $\eta_1$ is more sensitive to changes in the parameters $p$ and $\theta$ than $\eta_2$, indicating that the sell boundary is more affected by changes in the agent's risk aversion and liquidity preference. Therefore, to make informed investment decisions, it is important to consider how changes in both the liquid and illiquid assets affect the buy-sell boundaries

\subsection{\bf Optimal investment and consumption}
We now examine the investment ratio $\pi$ and consumption ratio $\frac{c}{x+y}$ in the no-trading region $NT$, focusing on the effect of the liquidity preference parameter $\theta$. 

To investigate this, we fix $\left(\alpha_1,\sigma_1\right)=\left(4\%,30\%\right)$, $\left(\alpha_2,\sigma_2\right)=(8\%,35\%)$, and consider values of $p$ as both positive ($p=0.3$) and negative ($p=-0.3$). We take $\theta=0.2, 0.4, 0.6, 0.8,$ and $1$. The impact of changes to the effective interval $\left(\eta_2,\eta_1\right)$ has been discussed earlier, so we will not address it here. 
\begin{figure}[htbp]
	\centering
	\includegraphics[scale=0.73]{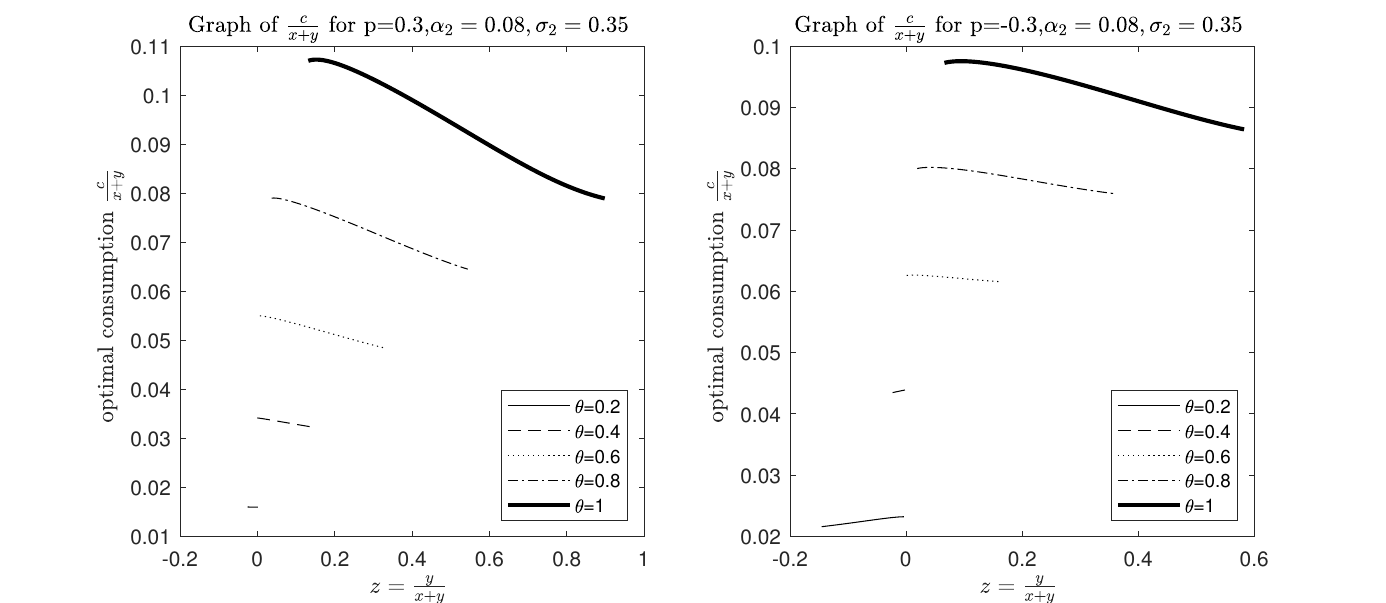}
	\caption{Graph of $\frac{c}{x+y}$, $\theta = 0.2, 0.4, 0.6, 0.8, 1$.}
	\label{tu1}
\end{figure}

Fig.~\ref{tu1} illustrates that the consumption ratio $\frac{c}{x+y}$ decreases as both the ratio of illiquid wealth $z$ and the liquidity preference parameter $\theta$ increase. This can be explained by the fact that a larger proportion of illiquid wealth reduces the amount of liquid wealth available for consumption. As a result, the consumption ratio decreases as $z$ increases. Similarly, a higher liquidity preference parameter $\theta$ indicates a greater preference for holding liquid wealth rather than spending it on consumption, leading to decreased consumption. Therefore, the decrease in the consumption ratio with increasing $\theta$ is consistent with the agent's stronger preference for holding onto liquid wealth rather than using it for consumption.
\begin{figure}[htbp]
	\centering
	\includegraphics[scale=0.73]{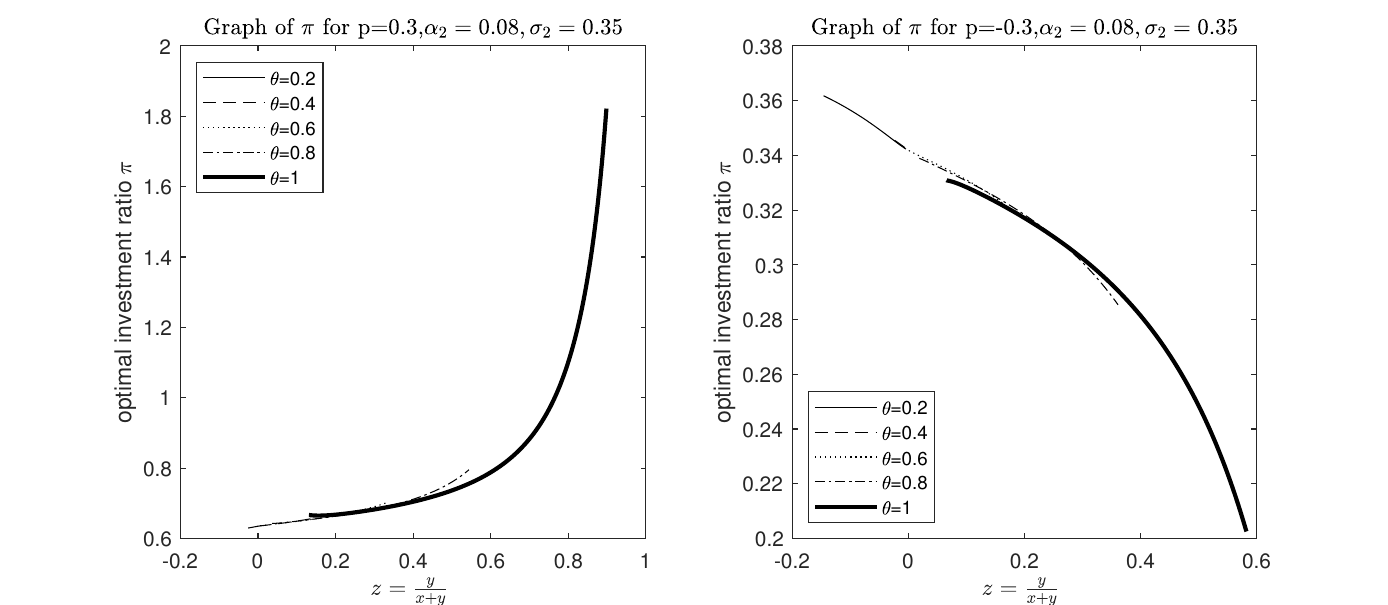}
	\caption{Graph of $\pi$, $\theta = 0.2, 0.4, 0.6, 0.8, 1$.}
	\label{tu2}
\end{figure}

Fig.~\ref{tu2} shows that the liquidity preference parameter $\theta$ mainly affects the investment interval $\left(\eta_2,\eta_1\right)$, but has little impact on the investment ratio $\pi$, which is also intuitive. Liquidity preferences balance the ratio between liquid and illiquid wealth, but makes no distinction within liquid wealth (including when modeling) and is therefore unlikely to affect investment within liquid wealth. For agents with low risk aversion ($p=0.3$), their investment ratio $\pi$ tends to increase as the proportion of illiquid assets $z$ increases. In contrast, for agents with high risk aversion ($p=-0.3$), the opposite is true: their investment ratio $\pi$ tends to decrease as the proportion of illiquid assets increases. This pattern can be explained by the fact that low-risk aversion agents increase their investment in liquid risk assets to balance the liquidity-illiquidity ratio when the proportion of illiquid assets increases, while high-risk aversion agents invest more in liquid risk assets to mitigate risk. Moreover, agents with higher risk aversion levels have smaller investment ratios $\pi$ overall.
\begin{figure}[htbp]
	\centering
	\includegraphics[scale=0.73]{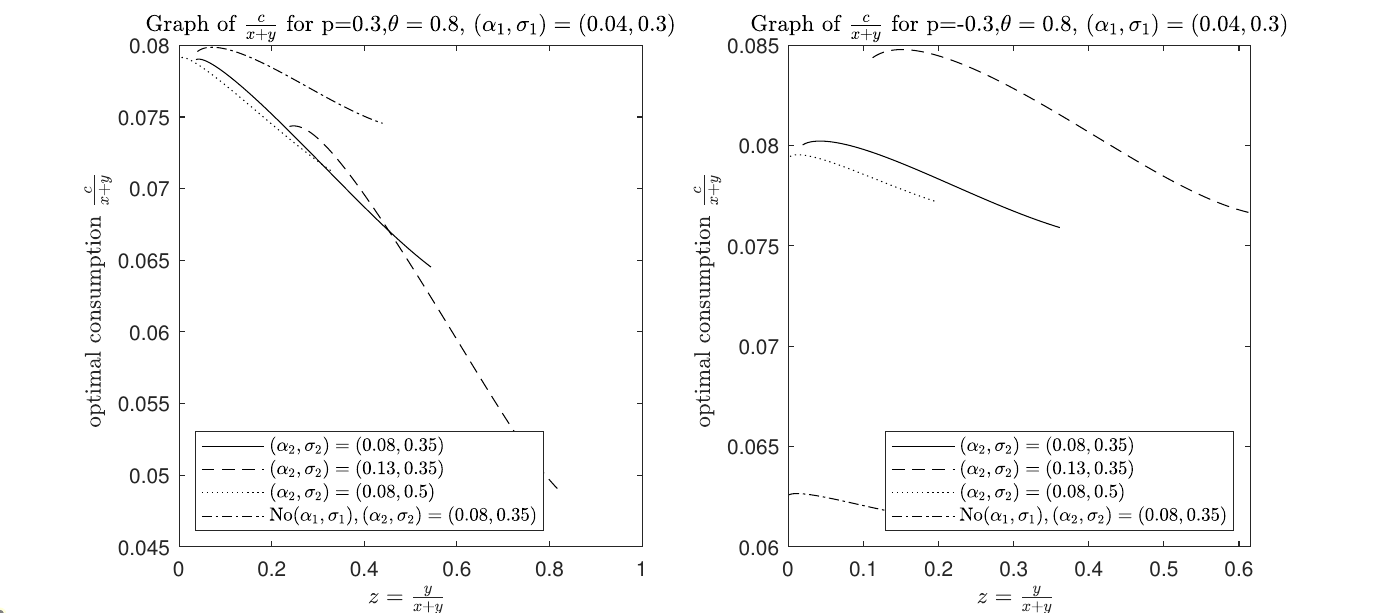}
	\caption{Graph of $\frac{c}{x+y}$, the same parameters as scenarios 1-4.}
	\label{tu3}
\end{figure}

To explore the impact of the parameters related to $S^1$ and $S^2$, we fix $\theta=0.8$, consider values of $p$ as both positive ($p=0.3$) and negative ($p=-0.3$), and take the same parameters as in scenarios 1-4.

Fig.~\ref{tu3} shows that $S^2$ with $\left(\alpha_2,\sigma_2\right)$ has little effect on the consumption ratio $\frac{c}{x+y}$ in the no-trading region $NT$ when $p=0.3$. However, the introduction of the liquid risk asset $S^1$ leads to a smaller consumption ratio for agents with low risk aversion and a larger consumption ratio for agents with high risk aversion, consistent with the previous analysis. This is because low-risk aversion agents seek to balance the liquidity-illiquidity ratio, while high-risk aversion agents aim to mitigate risks. Furthermore, Fig.~\ref{tu3} indicates that the optimal consumption ratio is higher when the expected return of the illiquid asset is larger, while a larger risk of the illiquid asset leads to a smaller optimal consumption ratio.
\begin{figure}[htbp]
	\centering
	\includegraphics[scale=0.73]{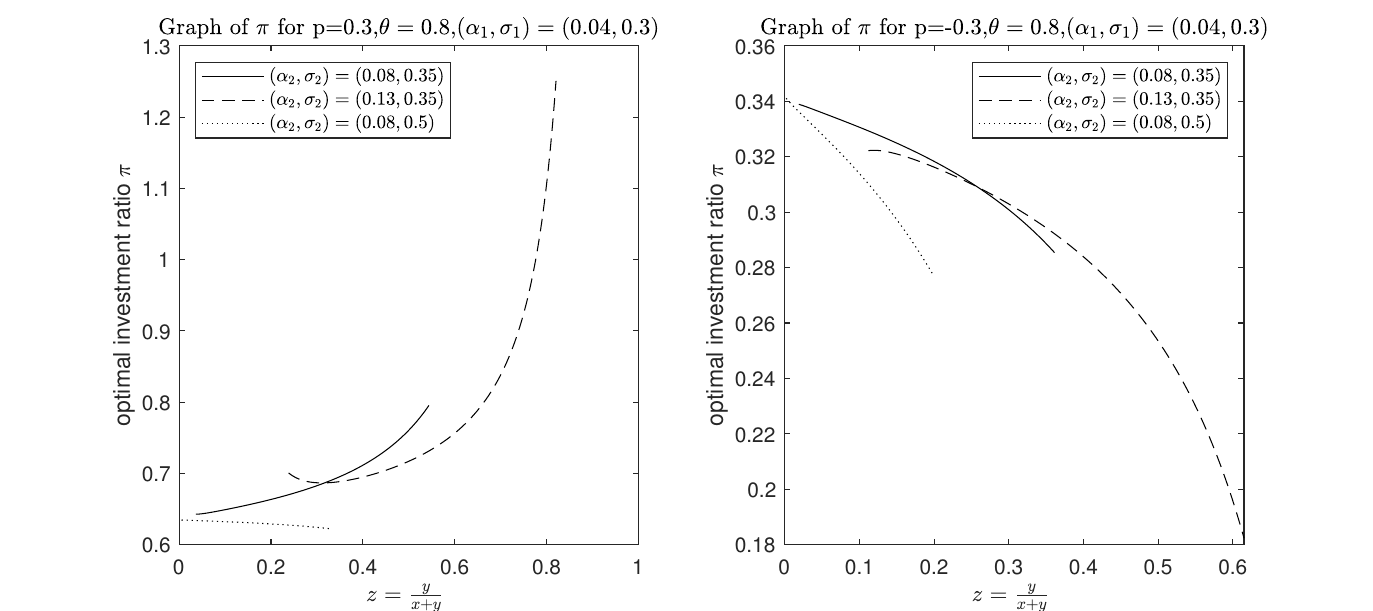}
	\caption{Graph of $\pi$, the same parameters as scenarios 1-4.}
	\label{tu4}
\end{figure}

Fig.~\ref{tu4} shows that the impact of the expected return $\alpha_2$ of the illiquid stock on the optimal investment ratio in the liquid stock is complex and depends on the agent's level of risk aversion and the proportion of illiquid assets $z$. For an agent with low risk aversion ($p=0.3$), the optimal investment ratio in the liquid stock increases with $\alpha_2$ when $z$ is small, indicating that a higher expected return on the illiquid stock leads to a higher investment in the liquid stock. However, when $z$ is large, the optimal investment ratio in the liquid stock decreases with $\alpha_2$, indicating that a higher expected return on the illiquid stock leads to a lower investment in the liquid stock. On the other hand, for an agent with high risk aversion ($p=-0.3$), the optimal investment ratio in the liquid stock decreases with $\alpha_2$ when $z$ is small, indicating that a higher expected return on the illiquid stock leads to a lower investment in the liquid stock. However, when $z$ is large, the optimal investment ratio in the liquid stock increases with $\alpha_2$, indicating that a higher expected return on the illiquid stock leads to a higher investment in the liquid stock. Furthermore, it is worth noting that the illiquid risk asset and the liquid risk asset are positively correlated, meaning that when the illiquid asset performs well, the liquid asset is also likely to perform well. Observing Fig.~\ref{tu4}, we see that the agent always decreases investment in liquid risk assets if the illiquid risk asset has a larger volatility coefficient $\sigma_2$. This can be explained by the fact that a higher volatility of the illiquid asset implies a higher overall risk of the portfolio, which makes the agent more cautious and leads to a smaller allocation to risk assets. In contrast, when the volatility of the illiquid asset is lower, the agent is more willing to take on risk and allocate more of their portfolio to the risk assets.

\section{\bf Conclusion}
This paper investigates an infinite horizon, discounted, consumption-portfolio problem in a market with a risk-free bond, a liquid risky asset, and an illiquid risky asset. The liquid wealth is composed of the risk-free bond and the liquid risky asset and it is introduced into the utility function to capture the agent's liquidity preference, while the illiquid risky asset requires proportional transaction costs when transacted. We define the HJB equation and the value function and study the properties of the value function. We prove that this problem is reachable and mathematically characterize the optimal policy, which can be accurately solved numerically. Despite the introduction of the liquid risk asset and the consideration of liquidity preference in the utility function, the optimal policy can still be divided into three conical regions (buying region, no-trading region, and selling region), and these three regions are non-empty regardless of what risk aversion parameters the agent has.

In our study, we conduct numerical analysis to investigate the impact of parameters on the location of the no-trading region and the investment and consumption levels above it. Our findings differ from previous studies, as we show that the illiquid risky asset may be negative in the no-trading region or even the selling region due to the introduction of liquidity preference and the liquid risk asset. We find that liquidity preference encourages agents to hold more liquid wealth and reduces consumption, but has little impact on the internal investment of liquid wealth. Additionally, the introduction of the liquid risk asset affects the location of the no-trading region, and its impact on consumption depends on the level of risk aversion of agents. Our results highlight the importance of considering liquidity preference and risk aversion in portfolio optimization, especially in markets with illiquid assets and transaction costs.

%On the one hand, we are the first to consider liquidity preference and use ``liquid wealth in the utility function" in consumption-portfolio problems with proportional transaction costs. And the influence of liquidity preference on consumption-portfolio decisions of liquid/illiquid assets is discussed. On the other hand, like \cite{Sherve1994}, we pay attention to the existence and location of three regions in the optimal policy. The optimal policy exists and can be described, in which the three regions are both non-empty.
\vskip 15pt
{\bf Acknowledgements.}The authors acknowledge the support from the National Natural Science Foundation of China (Grant No.12271290, No.11901574, No.11871036), and the MOE Project of Key Research Institute of Humanities and Social Sciences (22JJD910003). The authors thank the members of the group of Mathematical Finance and Actuarial Science at the Department of Mathematical Sciences, Tsinghua University for their {feedback} and useful conversations.

\appendix
\renewcommand{\theequation}{\thesection.\arabic{equation}}
\vskip 15pt

\bibliographystyle{apalike}
\bibColoredItems{red}{oecd}
\bibliography{wpref}
\end{document}